\documentclass[letterpaper, 11pt]{llncs}

\usepackage[margin=1in]{geometry}
\usepackage[bookmarks, colorlinks=true, plainpages = false, citecolor = blue,linkcolor=red,urlcolor = blue, filecolor = blue]{hyperref}
\usepackage{url}
\usepackage{amsmath,amsfonts,amssymb,bm, verbatim,dsfont,mathtools}
\usepackage{algorithm,algorithmic,color,graphicx,appendix}
\usepackage{subfigure}
\usepackage{etoolbox}
\usepackage{tikz}
\usepackage{xr,xspace}
\usepackage{todonotes}
\usepackage{paralist}
\usepackage{caption}
\newtheorem{assumption}{Assumption}

\newcommand{\eproof}{\hfill $\Box$}

\usepackage{xspace,prettyref}
\newcommand{\diverge}{\to\infty}

\newcommand{\ones}{\mathbf 1}

\newcommand{\reals}{{\mathbb{R}}}

\newcommand{\naturals}{{\mathbb{N}}}

\newcommand{\expect}[1]{\mathbb{E}\left[ #1 \right]}

\newcommand{\toas}{\xrightarrow{{\rm a.s.}}}

\newrefformat{eq}{(\ref{#1})}
\newrefformat{chap}{Chapter~\ref{#1}}
\newrefformat{sec}{Section~\ref{#1}}
\newrefformat{algo}{Algorithm~\ref{#1}}
\newrefformat{fig}{Fig.~\ref{#1}}
\newrefformat{tab}{Table~\ref{#1}}
\newrefformat{rmk}{Remark~\ref{#1}}
\newrefformat{clm}{Claim~\ref{#1}}
\newrefformat{def}{Definition~\ref{#1}}
\newrefformat{cor}{Corollary~\ref{#1}}
\newrefformat{lmm}{Lemma~\ref{#1}}
\newrefformat{prop}{Proposition~\ref{#1}}
\newrefformat{app}{Appendix~\ref{#1}}
\newrefformat{hyp}{Hypothesis~\ref{#1}}
\newrefformat{thm}{Theorem~\ref{#1}}

\newcommand{\pth}[1]{\left( #1 \right)}
\newcommand{\qth}[1]{\left[ #1 \right]}
\newcommand{\sth}[1]{\left\{ #1 \right\}}

\newcommand{\calC}{{\mathcal{C}}}
\newcommand{\calD}{{\mathcal{D}}}
\newcommand{\calE}{{\mathcal{E}}}

\newcommand{\calH}{{\mathcal{H}}}
\newcommand{\calI}{{\mathcal{I}}}

\newcommand{\calK}{{\mathcal{K}}}
\newcommand{\calL}{{\mathcal{L}}}

\newcommand{\calN}{{\mathcal{N}}}

\newcommand{\calR}{{\mathcal{R}}}
\newcommand{\calS}{{\mathcal{S}}}

\newcommand{\calV}{{\mathcal{V}}}

\begin{document}

\title{Asynchronous Distributed Hypothesis Testing\\ in the Presence of Crash Failures
\thanks{This research is supported in part by National Science Foundation award NSF 1329681. 
Any opinions, findings, and conclusions or recommendations expressed here are those of the authors
and do not necessarily reflect the views of the funding agencies or the U.S. government.}}

\author{Lili Su and Nitin H. Vaidya}
\institute{Department of Electrical and Computer Engineering, and\\
Coordinated Science Laboratory\\
University of Illinois at Urbana-Champaign\\
Email:\{lilisu3, nhv\} @illinois.edu}
\maketitle

\begin{center}
Technical Report\\
~\\

\today
\end{center}
~

\centerline{\bf Abstract}

This paper addresses the problem of distributed hypothesis testing in multi-agent networks, where
 agents repeatedly collect local observations about an {\em unknown} state of the world, and try to collaboratively detect the true state through information exchange.
 We focus on the impact of failures and asynchrony --  two fundamental factors in distributed systems -- on the performance of  consensus-based non-Bayesian learning. In particular, we consider the scenario where the networked agents may suffer crash faults, and messages delay can be {\em arbitrarily} long but finite.
We identify the minimal global detectability of the network for non-Bayesian rule to succeed. In addition, we obtain a generalization of a celebrated result by Wolfowitz and Hajnal to submatrices, which might be of independent interest.

\section{Introduction}
\label{intro}
Decentralized hypothesis testing is an important component of many decision-making and learning algorithms for large-scale systems, and thus has received significant amount of attention \cite{chamberland2003decentralized,gale2003bayesian,jadbabaie2012non,Tsitsiklis1988,tsitsiklis1993decentralized,varshney2012distributed,wong2012stochastic}. 
The traditional decentralized detection framework consists of a collection of spatially distributed sensors/agents and a fusion center \cite{Tsitsiklis1988,tsitsiklis1993decentralized,varshney2012distributed}. The sensors/agents independently collect {\em noisy} observations of the environment state, and send only {\em summary} of the private observations to the fusion center, where a final decision is made. 
In the case when the sensors/agents directly send all the private observations, the detection problem can be solved using a centralized scheme.
However, the above framework does not scale well, since  each sensor needs to be connected to the fusion center and full reliability of the fusion center is assumed, which may not be practical as the system scales.
Distributed hypothesis testing in the absence of fusion center is introduced by Gale and Kariv \cite{gale2003bayesian} in the context of social learning, where fully Bayesian belief update rule is studied.
 Bayesian update rule is impractical in many applications due to memory and computation constraints of each agent, and the inter-agent coordination challenges. 

Non-Bayesian learning rule is first proposed by Jadbabaie et al. \cite{jadbabaie2012non} in the general setting where external signals are observed during each iteration of the algorithm execution. Specifically, the belief of each agent is repeatedly updated as the arithmetic mean of its local Bayesian update and the beliefs of its neighbors, combing iterative consensus algorithm with local Bayesian update. It is shown \cite{jadbabaie2012non} that, under this learning rule,
each agent learns the true state almost surely.
%
Since the publication of \cite{jadbabaie2012non}, significant efforts have been devoted to designing and analyzing non-Bayesian learning rules with a particular focus on refining the fusion strategies and analyzing the (asymptotic and/or finite time) convergence rates of the refined algorithms \cite{jadbabaie2013information,nedic2014nonasymptotic,rad2010distributed,shahrampour2013exponentially,Lalitha2014,shahrampour2015finite,shahrampour2014distributed,molavi2015foundations}.

Among the various proposed fusion rules,  
in this paper we are particularly interested in the log-linear form of the update rule, in which, essentially, each agent updates its belief as the geometric average of the local Bayesian update and its neighbors' beliefs \cite{rad2010distributed}.
The log-linear form (geometric averaging) update rule is proposed in \cite{rad2010distributed}, and is shown to converge exponentially fast \cite{jadbabaie2013information,shahrampour2013exponentially}.
Taking an axiomatic approach, the geometric averaging fusion is shown to be  optimal \cite{molavi2015foundations}.
An optimization-based interpretation of this rule is presented in
\cite{shahrampour2013exponentially}, using dual averaging method with properly chosen proximal functions.
 Finite-time convergence rates are investigated independently in \cite{nedic2014nonasymptotic,Lalitha2014,shahrampour2014distributed}.
 Both \cite{nedic2014nonasymptotic} and \cite{shahrampour2015finite} consider time-varying networks, with slightly different network models. Specifically, \cite{nedic2014nonasymptotic} assumes that the union of every consecutive $B$ networks is strongly connected, while \cite{shahrampour2015finite} considers random networks.
In this paper, for ease of exposition, we assume that the network topology is static. As can be seen later, our results can be easily generalized to time-varying networks.

All the above work implicitly assumes synchronous systems and reliable agents.
However, in the context of decentralized hypothesis testing, asynchrony (message asynchrony and/or computation asynchrony) and failures (link failures and/or agent failures) --  two fundamental factors in practical distributed systems -- have not been addressed yet. In this paper, we consider the scenario where some agents may suffer crash faults (i.e., cease operating), and messages delay can be {\em arbitrarily} long but finite \cite{Lynch:1996:DA:525656}. This model is commonly considered in distributed computing, and the main challenge is that when an agent $i$ does not receive messages from the incoming neighbor agent $j$ within some timeout interval, from agent $i$'s perspective, it is not possible to distinguish whether agent $j$ has crashed, or the messages from agent $j$ are delayed.

We identify the minimal global detectability of the network for non-Bayesian rule to succeed. In addition, we obtain a generalization of a celebrated result by Wolfowitz and Hajnal to submatrices, which might be of independent interest.

The rest of the paper is organized as follows. Section \ref{prob formulation} presents the problem formulation and the learning rule.
Section \ref{matrix representation} introduces the notion of pseudo-belief vector, whose evolution captures the dynamics of the agents' beliefs, and admits a simple matrix representation. The limiting behavior of the product of the update matrices (the matrices in the matrix representation introduced in Section \ref{matrix representation}) is investigated in Section \ref{sec: consensus}. Also in Section \ref{sec: consensus}, we generalize a celebrated result by Wolfowitz and Hajnal.  The
necessary and sufficient network detectability is characterized in Section \ref{sec: almost sure}.

\section{Problem Formulation}
\label{prob formulation}
\paragraph{Network Model:}
We consider an asynchronous system, where the message delay can be {\em arbitrarily} long but finite.
A collection of $n$ agents are connected by a {\em directed} network $G(\calV, \calE)$, where $\calV=\{1, \ldots, n\}$ and $\calE$ is the collection of {\em directed} edges. For each $i\in \calV$, let $\calI_i$ denote the set of incoming neighbors of agent $i$. In any execution, up to $f$ agents suffer crash faults (i.e., cease operating). Let $\calN\subseteq \calV$ be the set of agents that operate correclty during a given execution (i.e., they are non-faulty). Note that $|\calV-\calN|\le f$, since at most $f$ agents may suffer crash failure.
As noted earlier, although we assume a static network topology, our results can be easily generalized to time-varying networks. Throughout this paper, we use the terms agent and node interchangeably.

\paragraph*{Observation Model:}
Let $\Theta=\{\theta_1, \theta_2, \ldots, \theta_m\}$ denote a finite set of $m$ environmental states, which we call {\em hypotheses}.
We consider asynchronous iterations. 
In the $t$-th iteration, each agent {\em independently} obtains
private signal about the environmental state $\theta^*$, which is initially unknown to every agent in the network.
For ease of exposition, we assume that if multiple signals are observed, only one signal is used to update beliefs.
Each agent $i$ knows the structure of its private signal, which is represented by a collection of parameterized distributions $\calD^i=\{\ell_i(w_i | \theta)| \theta\in \Theta,\, w_i\in \calS_i\}$,
where $\ell_i(\cdot | \theta)$ is a distribution with parameter $\theta\in \Theta$, and $\calS_i$ is the finite private signal space. For each $\theta \in \Theta$, and each $i\in \calV$, the support of  $\ell_i(\cdot|\theta)$ is the whole signal space, i.e., $\ell_i(w_i|\theta)>0$, $\forall\, w_i\in \calS_i$ and $\forall\, \theta \in \Theta$.
Precisely, let $s_t^i$ be the private signal observed by agent $i$ in iteration $t$, and let ${\bf s}_t=\{s_t^1, s_t^2, \ldots, s_t^n\}$ be the signal profile at time $t$ (i.e., signals observed by the agents in iteration $t$). Given an environmental state $\theta$, the signal profile ${\bf s}_t$ is generated according to the joint distribution $\ell_1(\cdot|\theta)\times \ell_2(\cdot|\theta)\times \cdots \times \ell_n(\cdot|\theta)$.
The goal is to have all the non-faulty agents (agents in $\calN$) collaboratively learn the environmental
state $\theta^*$. 

\paragraph{Belief Update Rule:}
Each agent $i$ keeps a belief $\mu^i$, which is a distribution over the set $\Theta$, with $\mu^i(\theta)$ being the probability with which the agent $i$ {\em believes} that $\theta$ is the true environmental state. Since no signals are observed before the execution of an algorithm, the belief $\mu^i$ is often initially set to be uniform over the set $\Theta$, i.e.,
$\pth{\mu_0^i(\theta_1), \mu_0^i(\theta_1), \ldots, \mu_0^i(\theta_m)}^T=\pth{\frac{1}{m}, \ldots, \frac{1}{m}}^T$. \footnote{In this paper, every vector considered is column vector.}
In this work, we also adopt the above convention. (For our results to hold, it suffices to have $\mu_0^i(\theta)>0$ for each $\theta\in \Theta$ and each $i\in \calV$.)

In our algorithm,
we will use a geometric averaging update rule that has been investigated in previous work
\cite{nedic2014nonasymptotic,rad2010distributed,Lalitha2014,shahrampour2014distributed}. 
Let $\calN[t]$ be the set of agents that have not crashed {\em by the beginning} of iteration $t$, and let $\bar{\calN}[t]$ be the set of agents that have not crashed {\em by the end} of iteration $t$. Note that $\calN[t+1]\subseteq \calN[t]$,  $\calN[t]-\bar{\calN}[t]$ is the collection of agents that crash {\em during} iteration $t$, and that $\bar{\calN}[t]=\calN[t+1]$. In addition, $\lim_{t\diverge} \calN[t]=\calN=\lim_{t\diverge}\bar{\calN}[t].$

For $t\ge 1$, the steps to be performed by agent $i\in \calN[t]$ in the $t$--th iteration are listed as follows, where messages are tagged with (asynchronous) iteration index.
%
%
%
%
%
%
%
%
\begin{enumerate}
\item {\em Transmit Step:} Transmit current belief vector $\mu_{t-1}^i$ on all outgoing edges.
\item {\em Receive Step:}  Wait until a private signal $s_t^i$ is observed and belief vectors are received from $|\calI_i|-f$ incoming neighbors. Let $\calR_i[t]$ be the set of incoming neighbors from whom agent $i$  receives these belief vectors.
\item {\em Update Step:}\footnote{In the notation $\mu_{t}^i(\theta)$, the superscript denotes the agents and subscript denotes iterations.} 
    For each $\theta\in \Theta$, update $\mu^i(\theta)$ as
\begin{align}
\label{learning update}
\mu_{t}^i(\theta)\triangleq  \frac{\ell_i(s_{t}^i|\theta)\prod_{j\in \calR_i[t]
\cup \{i\}}\mu_{t-1}^j(\theta)^{\frac{1}{|\calI_i|-f+1}}}{\sum_{p=1}^m \ell_i(s_{t}^i|\theta_p)\prod_{j\in \calR_i[t]
\cup \{i\}}\mu_{t-1}^j(\theta_p)^{\frac{1}{|\calI_i|-f+1}}}.
\end{align}
\end{enumerate}
%
%
%
%
Note that due to asynchrony and agent failures, $\calR_i[t]$ may change over time and {\em is not} monotone. In contrast, in synchronous systems, $\calR_i[t]$ is non-increasing, i.e., $\calR_i[t+1]\subseteq \calR_i[t]$ for any $t\ge 1$ and any $i\in \calV$.

In iteration $t$, if an agent crashes after performing the update step in (\ref{learning update}) for all $\theta\in \Theta$, without loss of generality, we say that this agent crashes in iteration $t+1$.
Note that each agent in $\calN[t]-\bar{\calN}[t]$  may crash at any time during iteration $t$. In particular, it may crash before performing \eqref{learning update} or while performing \eqref{learning update}.
The above protocol is different from the original geometric averaging learning rule \cite{nedic2014nonasymptotic,rad2010distributed,Lalitha2014,shahrampour2014distributed} in the receive step, where each agent waits to receive messages from $|\calI_i|-f$ incoming neighbors instead of waiting to hear from all of its incoming neighbors. This modification is necessary because of asynchrony and the possibility of up to $f$ agents crashes. Requiring an agent to receive messages from more than $|\calI_i|-f$ neighbors may lead to non-termination of the learning protocol (when $f$ incoming neighbors have crashed).

Recall that $\theta^*$ is the true state.
We say the networked agents collaboratively detect $\theta^*$ if for each non-faulty agent $i\in \calN$,
\begin{align}
\label{goal}
\lim_{t\diverge}\mu_t^i(\theta^*) ~=~ 1 \,\,\, a.s.
\end{align}

\section{Matrix Representation}
\label{matrix representation}
In this section, we define a matrix representation of the agents' belief update. In synchronous and reliable networks, the update matrix ${\bf A}[t]$ is often chosen to be the weighted adjacency matrix of the network \cite{nedic2014nonasymptotic,Lalitha2014,shahrampour2014distributed}.
However, the above matrix representation is improper in the presence of message asynchrony and agent failure, observing that transmitted messages may not be used due to message delay, and agents {\em not} in $\bar{\calN}[t]$ may not perform  \eqref{learning update} for all $\theta\in \Theta$  (and may not even observe new private signals).
To resolve the above complication, we introduce pseudo-belief vector $\tilde{\mu}^i$ for each $i\in \calV$.
The evolution of the pseudo-belief vectors admits a simple matrix representation, and captures the update of $\mu^i$ for each non-faulty agent $i\in \calN$.
For each $i\in \bar{\calN}[t]$, define
\begin{align}
\label{matrix 1}
{\bf A}_{ij}[t]\triangleq
\begin{cases}
\frac{1}{|\calI_i|-f+1}, & j\in \calR_i[t]\cup \{i\}\\
0, & \text{otherwise,}
\end{cases}
\end{align}
and for each $i\notin \bar{\calN}[t]$,
\begin{align}
\label{matrix 2}
{\bf A}_{ij}[t]\triangleq
\begin{cases}
1, & j=i\\
0, & \text{otherwise.}
\end{cases}
\end{align}
Note that in iteration $t$, an agent can only receive messages tagged with (asynchronous) iteration index $t$ from agents that have not crashed by the beginning of iteration $t$. Thus,
${\bf A}_{ij}[t]=0$ for each $j\notin \calN[t]$, and
$$
1=\sum_{j=1}^n {\bf A}_{ij}[t] = \sum_{j\in \calN[t]} {\bf A}_{ij}[t]+\sum_{j\notin \calN[t]} {\bf A}_{ij}[t]=\sum_{j\in \calN[t]} {\bf A}_{ij}[t].
$$

The pseudo-belief is defined as follows.
For each $i\in \bar{\calN}[t]$,
\begin{align}
\label{auc1}
\tilde{\mu}_t^i(\theta)\triangleq \frac{\ell_i(s_{t}^i|\theta)\prod_{j=1}^n \tilde{\mu}_{t-1}^j(\theta)^{{\bf A}_{ij}[t]}}{\sum_{p=1}^m \ell_i(s_{t}^i|\theta_p)\prod_{j=1}^n \tilde{\mu}_{t-1}^j(\theta_p)^{{\bf A}_{ij}[t]}}, ~\forall \, \theta\in \Theta ,
\end{align}
with $\tilde{\mu}_0^i(\theta)=\mu_0^i(\theta)=\frac{1}{m}$, and for each $i\notin \bar{\calN}[t]$,
\begin{align}
\label{auc2}
\tilde{\mu}_t^i(\theta)~\triangleq ~\frac{\prod_{j=1}^n \tilde{\mu}_{t-1}^j(\theta)^{{\bf A}_{ij}[t]}}{\sum_{p=1}^m\prod_{j=1}^n \tilde{\mu}_{t-1}^j(\theta_p)^{{\bf A}_{ij}[t]}}=\tilde{\mu}_{t-1}^i(\theta), ~\forall \, \theta\in \Theta ,
\end{align}
with $\tilde{\mu}_0^i(\theta)=\mu_0^i(\theta)=\frac{1}{m}$.
Note that, in contrast to \eqref{auc1}, no private signal is involved in \eqref{auc2}. In addition, \eqref{auc2} is equivalent to defining $\tilde{\mu}_t^i(\theta)=\tilde{\mu}_{t-1}^i(\theta)$
It is easy to see (by induction) that
$$\tilde{\mu}_t^i(\theta)=\mu_t^i(\theta), \, \forall\, i\in \bar{\calN}[t].$$
%
%
%
Note that $\tilde{\mu}^i(\theta)$ only captures the evolution of $\mu^i(\theta)$ for $i\in \bar{\calN}[t]$, i.e., each agent that has not crashed {\em at the end} of iteration $t$. Since an agent $i\in \calN[t]-\bar{\calN}[t]$ may crash during the update step (\ref{learning update}), $\tilde{\mu}^i(\theta)$ may not capture the real update performed by nodes in $\calN[t]-\bar{\calN}[t]$ for some subset of $\Theta$. This inconsistency does not affect the accuracy of our analysis. Intuitively, since the nodes in $\calN[t]-\bar{\calN}[t]$ are crashing away, they will not affect further system evolution.

For any $\theta\in \Theta$, and any $i\in \calV$, let
\begin{align}
\label{b1}
\bm{\psi}_{t}^i(\theta)\triangleq \log \frac{\tilde{\mu}_t^i(\theta)}{\tilde{\mu}_t^i(\theta^*)}.
\end{align}
Note that for each $i\in \calN$, $\mu_t^i(\theta^*)\toas 1$ if and only if $\bm{\psi}_{t}^i(\theta)\toas -\infty$ for $\theta\not=\theta^*$.
In addition, let $\calL_{t}(\theta)\in \reals^n$ such that
\begin{align}
\label{private signal}
\calL^i_{t}(\theta)~\triangleq~
\begin{cases}
\log \frac{\ell_i(s_{t}^i|\theta)}{\ell_i(s_{t}^i|\theta^*)}, &~ ~\forall \, i\in \bar{\calN}[t],\\
0, &~ \text{otherwise}.
\end{cases}
\end{align}
Then, for each $i\in \bar{\calN}[t]$, we have
\begin{align}
\label{rewrite1}
\nonumber
\bm{\psi}_{t}^i(\theta)&=\log \frac{\tilde{\mu}_{t}^i(\theta)}{\tilde{\mu}_{t}^i(\theta^*)}\\
\nonumber
&=\log \pth{ \frac{\pth{\prod_{j=1}^n\tilde{\mu}_{t-1}^j(\theta)^{{\bf A}_{ij}[t]}}\ell_i(s_{t}^i|\theta)}{\sum_{p=1}^m \prod_{j=1}^n\tilde{\mu}_{t-1}^j(\theta)^{{\bf A}_{ij}[t]}\ell_i(s_{t}^i|\theta)}\times \frac{\sum_{p=1}^m \prod_{j=1}^n\tilde{\mu}_{t-1}^j(\theta)^{{\bf A}_{ij}[t]}\ell_i(s_{t}^i|\theta)}{\pth{\prod_{j=1}^n\tilde{\mu}_{t-1}^j(\theta^*)^{{\bf A}_{ij}[t]}}\ell_i(s_{t}^i|\theta^*)}}\\
\nonumber
&=\log \pth{\prod_{j=1}^n \pth{\frac{\tilde{\mu}_{t-1}^j(\theta)}{\tilde{\mu}_{t-1}^j(\theta^*)}}^{{\bf A}_{ij}[t]}\times \frac{\ell_i(s_{t}^i|\theta)}{\ell_i(s_{t}^i|\theta^*)}}\\
\nonumber
&=\sum_{j=1}^n {\bf A}_{ij}[t] \log \frac{\tilde{\mu}_{t-1}^j(\theta)}{\tilde{\mu}_{t-1}^j(\theta^*)} +\log \frac{\ell_i(s_{t}^i|\theta)}{\ell_i(s_{t}^i|\theta^*)}\\
&=\sum_{j=1}^n {\bf A}_{ij}[t] \bm{\psi}_{t-1}^j(\theta) +\calL^i_{t}(\theta) ~~~\text{by \eqref{private signal}}
\end{align}
and for each $i\notin \bar{\calN}[t]$, we have
\begin{align}
\label{rewrite2}
\nonumber
\bm{\psi}_{t}^i(\theta)&=\log \frac{\tilde{\mu}_{t}^i(\theta)}{\tilde{\mu}_{t}^i(\theta^*)}\\
\nonumber
&=\log \pth{ \frac{\pth{\prod_{j=1}^n\tilde{\mu}_{t-1}^j(\theta)^{{\bf A}_{ij}[t]}}}{\sum_{p=1}^m \prod_{j=1}^n\tilde{\mu}_{t-1}^j(\theta)^{{\bf A}_{ij}[t]}}\times \frac{\sum_{p=1}^m \prod_{j=1}^n\tilde{\mu}_{t-1}^j(\theta)^{{\bf A}_{ij}[t]}}{\pth{\prod_{j=1}^n\tilde{\mu}_{t-1}^j(\theta^*)^{{\bf A}_{ij}[t]}}}}~~\text{by \eqref{auc2}}\\
\nonumber
&=\log \pth{\prod_{j=1}^n \pth{\frac{\tilde{\mu}_{t-1}^j(\theta)}{\tilde{\mu}_{t-1}^j(\theta^*)}}^{{\bf A}_{ij}[t]}}\\
\nonumber
&=\sum_{j=1}^n {\bf A}_{ij}[t] \log \frac{\tilde{\mu}_{t-1}^j(\theta)}{\tilde{\mu}_{t-1}^j(\theta^*)}\\
&=\sum_{j=1}^n {\bf A}_{ij}[t] \bm{\psi}_{t-1}^j(\theta)+\calL^i_{t}(\theta) ~~~\text{by \eqref{private signal}}
\end{align}
Let $\bm{\psi}_{t}(\theta)\in \reals^n$ be the vector that stacks $\bm{\psi}_t^i(\theta)$, with the $i$--th entry being $\bm{\psi}_t^i(\theta)$ for all $i\in \calV$.
The evolution of $\bm{\psi}(\theta)$ can be compactly written as
\begin{align}
\label{matrix form}
\bm{\psi}_{t}(\theta)&={\bf A}[t]\bm{\psi}_{t-1}(\theta)+\calL_{t}(\theta).
\end{align}
Expanding \eqref{matrix form}, we get
\begin{align}
\label{int1}
\nonumber
\bm{\psi}_{t}(\theta)&={\bf A}[t]\bm{\psi}_{t-1}(\theta)+\calL_{t}(\theta)\\
\nonumber
&={\bf A}[t]\pth{{\bf A}[t-1]\bm{\psi}_{t-2}(\theta)+\calL_{t-1}(\theta)}+\calL_{t}(\theta) \\
\nonumber
&\cdots \\
\nonumber
&={\bf A}[t]{\bf A}[t-1]\cdots {\bf A}[1] \bm{\psi}_0(\theta)+\sum_{r=1}^{t-1} {\bf A}[t]{\bf A}[t-1]\cdots {\bf A}[r+1] \calL_r(\theta)+\calL_{t}(\theta)\\
&={\bf \Phi}(t,1)\bm{\psi}_0(\theta)+\sum_{r=1}^{t} {\bf \Phi}(t,r+1)\calL_r(\theta),
\end{align}
where ${\bf \Phi}(t, r)={\bf A}[t]\cdots {\bf A}[r]$ for $r\le t+1$. By convention, ${\bf \Phi}(t, t)={\bf A}[t]$ and ${\bf \Phi}(t, t+1)={\bf I}$.

\section{Convergence of ${\bf \Phi}(t, r)$}
\label{sec: consensus}
In this section, we present the tight (necessary and sufficient) condition on $G(\calV, \calE)$ \cite{tseng2015thesis} such that for $i,j\in \calN$, the following holds.
\begin{align}
\label{r weak ergodic}
\lim_{t\diverge}~\left | {\bf \Phi}_{ik}(t,r)-{\bf \Phi}_{jk}(t,r)\right |~=~0.
\end{align}
Note that \eqref{r weak ergodic} is weaker than requiring ${\bf \Phi}(t, r)$ to be weakly ergodic, where  \eqref{r weak ergodic} needs to hold for any $i,j \in \calV$. If the infinite backward product $\lim_{t\diverge}{\bf \Phi}(t, r)$ satisfies \eqref{r weak ergodic}, we say the product is weakly ergodic restricted to indices in $\calN$.
\begin{definition}
\label{reduced graph}
A reduced graph $\calH$ of $G(\calV, \calE)$ is obtained by (i) removing up to $f$ incoming links for each $i\in \calV$;
(ii) in the obtained graph, removing up to $f$ sinks, if any. \footnote{For a given graph $H(\calV, \calE)$, a node $s\in \calV$ is called a sink if there is no outgoing links coming from node $s$.}.
\end{definition}

Let $\calC$ be the collection of all reduced graph of $G(\calV, \calE)$.
Since $|\calV|=n$ is finite, it holds that $\chi=|\calC|<\infty$.

\begin{definition}
Given a graph $G(\calV, \calE)$,  a source component in  $G(\calV, \calE)$ is the strongly connected component that does not have an incoming link from outside the component.
\end{definition}
\noindent\text{\bf Condition 1:} Every reduced graph of $G(\calV, \calE)$ contains a unique source component.
\vskip 0.5\baselineskip

Let $\calH \in \calC$ be an arbitrary reduced of $G(\calV, \calE)$ with source component $\calS_{\calH}$. Define
\begin{align}
\label{minimal source}
\gamma\triangleq \min_{\calH\in \calC} |\calS_{\calH}|,
\end{align}
i.e., $\gamma$ is the minimum source component size in all reduced graphs. Note that $\gamma\ge 1$ if Condition 1 holds.

We first show that Condition 1 is necessary for \eqref{r weak ergodic} to hold.
Note that if \eqref{r weak ergodic} holds, asynchronous crash consensus is achieved \cite{Lynch:1996:DA:525656}. Thus the necessary condition for asynchronous consensus is also necessary for \eqref{r weak ergodic}. The following tight (necessary and sufficient) condition of asynchronous crash consensus is found in \cite{tseng2015thesis}.
\vskip 0.4\baselineskip
\noindent\text{\bf Condition 2:} For any node partition $L, R, C$ such that $L\not=\O$ and $R\not=\O$, at least one of the following holds: \\
(1) there exists $i\in L$ such that $|\calI_i\cap \pth{R\cup C}|\ge f+1$; or \\
(2) there exists $j\in R$ such that $|\calI_j\cap \pth{L\cup C}|\ge f+1$.

\vskip 0.5\baselineskip

\begin{theorem}
\label{equivalence}
Condition 1 and Condition 2 are equivalent.
\end{theorem}
Theorem \ref{equivalence} is proved in Appendix \ref{app:equivalence}.
\begin{remark}
Let $\calH^{\prime}$ be the subgraph of $G(\calV, \calE)$ obtained by removing up to $f$ incoming links for each $i\in \calV$. Indeed, we are able to show that Condition 2 holds if and only if $\calH^{\prime}$ contains a single source component. That is the removal of sink nodes does not matter for the convergence of consensus because it does not affect the uniqueness of source component and the only state evolution of the non-faulty nodes matter in the consensus problem.
\end{remark}
The sufficiency of Condition 1 is implied by several auxiliary propositions and lemmas. We first observe that the belief evolution is governed by reduced graphs.
\begin{proposition}
\label{lower bound}
Let $\xi=\frac{1}{1+\max_{i\in \calV} |\calI_i|}$. For $t\ge 1$, there exists a reduced graph $\calH[t]$ with adjacency matrix ${\bf H}[t]$ such that
$$
{\bf A}[t]\ge \xi {\bf H}[t].
$$ 
\end{proposition}
The proof of Proposition \ref{lower bound} can be found in Appendix \ref{app:lower bound}.

Let $\calH$ be a reduced graph of $G(\calV, \calE)$ with ${\bf H}$ as the adjacency matrix.
The convergence of the backward product ${\bf \Phi}(t, r)$ can be analyzed using ergodic coefficients. We generalize some well-known results obtained by Wolfowitz \cite{wolfowitz1963products} and Hajnal \cite{Hajnal58} to submatrices.

For $t\ge 1$, $t^{\prime}\ge t$ and $r\ge 1$, define
\begin{align}
\delta_{r}({\bf \Phi}(t^{\prime}, t))&\triangleq \max_{j\in \calV} \max_{i, i^{\prime} \in \calN[r]} \left | {\bf \Phi}_{i j}(t^{\prime}, t)-{\bf \Phi}_{i^{\prime} j}(t^{\prime}, t)\right |.
\label{ergodic 1}\\
\eta_{r}({\bf \Phi}(t^{\prime}, t)) &\triangleq \min_{i, i^{\prime} \in \calN[r]} \sum_{j\in \calV} \min \{{\bf \Phi}_{i j}(t^{\prime}, t), {\bf \Phi}_{i^{\prime} j}(t^{\prime}, t)\} \label{ergodic 2}.
\end{align}
Since $\calN[t]\subseteq \calN[r]$ for any $r\le t$, by \eqref{ergodic 1} and \eqref{ergodic 2}, we get
it is easy to see that
\begin{align}
\label{monotone 1}
\delta_{t}({\bf \Phi}(t^{\prime}, t))\le \delta_{r}({\bf \Phi}(t^{\prime}, t)),~\text{and}~~\eta_{t}({\bf \Phi}(t^{\prime}, t))\ge \eta_{r}({\bf \Phi}(t^{\prime}, t)).
\end{align}
Note that the definition of $\delta_r\pth{\bf \Phi}\pth{t^{\prime}, t}$ is symmetric in $i$ and $i^{\prime}$. Thus,
\begin{align}
\label{ergodic 11}
\nonumber
\delta_{r}({\bf \Phi}(t^{\prime}, t))&=\max_{j\in \calV} \max_{i, i^{\prime} \in \calN[r]} \left | {\bf \Phi}_{i j}(t^{\prime}, t)-{\bf \Phi}_{i^{\prime} j}(t^{\prime}, t)\right |\\
&=\max_{j\in \calV} \max_{i, i^{\prime} \in \calN[r]} \pth{ {\bf \Phi}_{i j}(t^{\prime}, t)-{\bf \Phi}_{i^{\prime} j}(t^{\prime}, t)}.
\end{align}
\begin{proposition}
\label{p3}
Given $t\ge 1$, for $t^{\prime}\ge t$, $i\in \calN[t]$ and $j\notin \calN[t]$, it holds that
\begin{align*}
{\bf \Phi}_{ij}(t^{\prime}, t)=0.
\end{align*}
\end{proposition}
\begin{proof}
We prove this proposition by inducting on $t^{\prime}$.

(Base case) When $t^{\prime}=t$, ${\bf \Phi}(t^{\prime}, t)={\bf A}[t]$. 
Since $i\in \calN[t]$ and $j\notin \calN[t]$, it holds that $i\not=j$, and ${\bf A}_{ij}[t]=0$, proving Proposition \ref{p3} for ${\bf \Phi}(t, t)$.

(Induction hypothesis) Suppose Proposition \ref{p3} holds for $t^{\prime}\ge t$.

(Inductive step) Now we prove Proposition \ref{p3} for ${\bf \Phi}(t^{\prime}+1, t)$.
Recall that ${\bf \Phi}(t^{\prime}+1, t)={\bf A}[t^{\prime}+1]{\bf \Phi}(t^{\prime}, t)$. For $i\in \calN[t]$ and $j\notin \calN[t]$, we get
\begin{align}
\label{induction}
\nonumber
{\bf \Phi}_{ij}\pth{t^{\prime}+1, t}&=\sum_{k=1}^n {\bf A}_{ik}[t^{\prime}+1]{\bf \Phi}_{kj}(t^{\prime}, t)\\
\nonumber
&=\sum_{k\in \calN[t]} {\bf A}_{ik}[t^{\prime}+1]{\bf \Phi}_{kj}(t^{\prime}, t)+\sum_{k\notin \calN[t]} {\bf A}_{ik}[t^{\prime}+1]{\bf \Phi}_{kj}(t^{\prime}, t)\\
\nonumber
&=\sum_{k\in \calN[t]} {\bf A}_{ik}[t^{\prime}+1]\,0+\sum_{k\notin \calN[t]} {\bf A}_{ik}[t^{\prime}+1]{\bf \Phi}_{kj}(t^{\prime}, t)~~~\text{by induction hypothesis}\\
&\le\sum_{k\notin \calN[t^{\prime}+1]} {\bf A}_{ik}[t^{\prime}+1]{\bf \Phi}_{kj}(t^{\prime}, t)~~~\text{since } \calN[t^{\prime}+1]\subseteq \calN[t].
\end{align}
For $k\not=i$, by \eqref{matrix 2}, we know ${\bf A}_{ik}[t^{\prime}+1]=0$. For $k=i$, by induction hypothesis, we know that
${\bf \Phi}_{ij}(t^{\prime}, t)=0$. Thus, we have
\begin{align*}
\sum_{k\notin \calN[t^{\prime}+1]} {\bf A}_{ik}[t^{\prime}+1]{\bf \Phi}_{kj}(t^{\prime}, t)&={\bf A}_{ii}[t^{\prime}+1]{\bf \Phi}_{ij}(t^{\prime}, t)+\sum_{k\notin \calN[t^{\prime}+1],\, k\not=i} {\bf A}_{ik}[t^{\prime}+1]{\bf \Phi}_{kj}(t^{\prime}, t) \\
&= 1\times 0+\sum_{k\notin \calN[t^{\prime}+1],\, k\not=i} 0\times{\bf \Phi}_{kj}(t^{\prime}, t)~=~0,
\end{align*}
i.e., the RHS of \eqref{induction} equals 0. This completes the proof of the induction, and Proposition \ref{p3} is proved.

\eproof
\end{proof}
Recall that ${\bf \Phi}(t^{\prime}, t)$ is row stochastic. As an immediate consequence of Proposition \ref{p3}, for any $t$, $t^{\prime}\ge t$, and $i\in \calN[t]$, the following holds.
\begin{align}
\label{sum}
\sum_{j\in \calN[t]} {\bf \Phi}_{ij}(t^{\prime}, t)=1.
\end{align}
The following two lemmas generalize the results from Wolfowitz \cite{wolfowitz1963products} and Hajnal \cite{Hajnal58}.
\begin{lemma}
\label{crash erdoc c1}
For each $t\ge 0$ and $t^{\prime}\ge t$, we have
\begin{align*}
\delta_{t}({\bf \Phi}(t^{\prime}, t))\le 1-\eta_{t}({\bf \Phi}(t^{\prime}, t)).
\end{align*}
\end{lemma}
\begin{proof}
\begin{align*}
&1-\eta_{t}({\bf \Phi}(t^{\prime}, t))=1-\min_{i, i^{\prime} \in \calN[t]} \sum_{j\in \calV} \min \{{\bf \Phi}_{i j }(t^{\prime}, t), {\bf \Phi}_{i^{\prime} j}(t^{\prime}, t)\}~~~\text{by (\ref{ergodic 2})}\\
&=\max_{i,i^{\prime}\in \calN[t]}\pth{1-\sum_{j\in \calV} \min \{{\bf \Phi}_{i j }(t^{\prime}, t), {\bf \Phi}_{i^{\prime} j}(t^{\prime}, t)\}}\\
&=\max_{i,i^{\prime}\in \calN[t]}\pth{\sum_{j\in \calV} {\bf \Phi}_{i j }(t^{\prime},t)-\sum_{j\in \calV} \min \{{\bf \Phi}_{i j }(t^{\prime}, t), {\bf \Phi}_{i^{\prime} j}(t^{\prime}, t)\}} ~\text{since }\sum_{j\in \calV} {\bf \Phi}_{i j }(t^{\prime},t)=1\\
&=\max_{i,i^{\prime}\in \calN[t]}\sum_{j: ~j\in \calV, {\bf \Phi}_{i j }(t^{\prime}, t) \ge {\bf \Phi}_{i^{\prime} j}(t^{\prime}, t) } \pth{{\bf \Phi}_{i j }(t^{\prime}, t)- {\bf \Phi}_{i^{\prime} j}(t^{\prime}, t)}\\
&\ge \max_{i,i^{\prime}\in \calN[t]}\max_{j\in \calV}\pth{ {\bf \Phi}_{i j }(t^{\prime}, t)- {\bf \Phi}_{i^{\prime} j}(t^{\prime}, t)}\\
&=\delta_t({\bf \Phi}(t^{\prime},t)) ~~~\text{by (\ref{ergodic 11})}.
\end{align*}

\end{proof}
\begin{lemma}
\label{crash erdoc c2}
 For $t_2> t_1\ge t_0\ge 1$,  define ${\bf P}={\bf \Phi}(t_2, t_1+1)$, ${\bf G}={\bf \Phi}(t_1, t_0)$, and ${\bf F}={\bf \Phi}(t_2, t_0).$
Then it holds that
\begin{align*}
\delta_{t_1+1}({\bf F}) \le \pth{1-\eta_{t_1+1}({\bf P})}\delta_{t_1+1}({\bf G}).
\end{align*}
\end{lemma}
\begin{proof}
Since ${\bf P}={\bf \Phi}(t_2, t_1+1)$, ${\bf G}={\bf \Phi}(t_1, t_0)$, and ${\bf F}={\bf \Phi}(t_2, t_0)$, it holds that
$${\bf F}={\bf \Phi}(t_2, t_0)={\bf \Phi}(t_2, t_1+1){\bf \Phi}(t_1, t_0)={\bf P}{\bf G}.$$
For any $i, j\in \calV$, we have ${\bf F}_{i j }=\sum_{k=1}^n{\bf P}_{i k}{\bf G}_{k j}$. We get
\begin{align}
\nonumber
\delta_{t_1+1}({\bf F})&= \max_{j \in \calV} \max_{i, i^{\prime}\in \calN[t_1+1]} \pth{{\bf F}_{i j }- {\bf F}_{i^{\prime} j}}~~~\text{by (\ref{ergodic 11})}\\
\nonumber
&=\max_{j \in \calV} \max_{i, i^{\prime}\in \calN[t_1+1]} \pth{\sum_{k=1}^n{\bf P}_{i k}{\bf G}_{k j}- \sum_{k=1}^n{\bf P}_{i^{\prime} k}{\bf G}_{k j}}\\
\nonumber
&=\max_{j \in \calV} \max_{i, i^{\prime}\in \calN[t_1+1]} \pth{\sum_{k=1}^n\pth{{\bf P}_{i k}- {\bf P}_{i^{\prime} k}}{\bf G}_{k j}}\\
&=\max_{j \in \calV} \max_{i, i^{\prime}\in \calN[t_1+1]} \pth{\sum_{k\in \calN[t_1+1]}\pth{{\bf P}_{i k}- {\bf P}_{i^{\prime} k}}{\bf G}_{k j}+ \sum_{k\notin \calN[t_1+1]}\pth{{\bf P}_{i k}- {\bf P}_{i^{\prime} k}}{\bf G}_{k j}} \label{partition1}.
\end{align}
%
%
%
Recall that $\calN[t_1+1]$ is the collection of agents that have not crashed by the beginning of iteration $t_1+1$. If $k\notin \calN[t_1+1]$, then $k$ crashed in the first $t_1$ iterations.  Intuitively speaking, for any $t>t_1$, since agent $k$ has already crashed, it cannot influence any other agents -- no other agents can receive any message from agent $k$ after iteration $t_1$. From Proposition \ref{p3}, we know that ${\bf P}_{i k}=0={\bf P}_{i^{\prime} k}$ for all $k\notin \calN[t_1+1]$, and $i, i^{\prime} \in \calN[t_1+1]$.
Consequently, we have
\begin{align}
\label{coeff aux 1}
\sum_{k\notin \calN[t_1+1]}\pth{{\bf P}_{i k}- {\bf P}_{i^{\prime} k}}{\bf G}_{k j}~=~\sum_{k\notin \calN[t_1+1]}0\, {\bf G}_{k j}~=~0, ~\forall\, j\in \calV.
\end{align}
Thus, equality (\ref{partition1}) can be simplified as
\begin{align}
\label{coeff aux}
\nonumber
\delta_{t_1+1}({\bf F})~&=\max_{j \in \calV} \max_{i, i^{\prime}\in \calN[t_1+1]} \pth{\sum_{k\in \calN[t_1+1]}\pth{{\bf P}_{i k}- {\bf P}_{i^{\prime} k}}{\bf G}_{k j}}~~~\text{by \eqref{coeff aux 1}}\\
\nonumber
&= \max_{j \in \calV} \max_{i, i^{\prime}\in \calN[t_1+1]} \Bigg{(}\sum_{k:~k\in \calN[t_1+1], \,{\bf P}_{ik}\ge {\bf P}_{i^{\prime}k}}\pth{{\bf P}_{i k}- {\bf P}_{i^{\prime} k}}{\bf G}_{k j}\\
\nonumber
&\quad \quad \quad \quad \quad \quad \quad \quad \quad \quad \quad \quad+\sum_{k:~k\in \calN[t_1+1], \,{\bf P}_{ik}< {\bf P}_{i^{\prime}k}}\pth{{\bf P}_{i k}- {\bf P}_{i^{\prime} k}}{\bf G}_{k j} \Bigg{)}\\
\nonumber
&\le \max_{j \in \calV} \max_{i, i^{\prime}\in \calN[t_1+1]} \Bigg{(}\sum_{k:~k\in \calN[t_1+1], \,{\bf P}_{i\,k}\ge {\bf P}_{i^{\prime}\,k}}\pth{{\bf P}_{i k}- {\bf P}_{i^{\prime} k}}\max_{k\in \calN[t_1+1]}{\bf G}_{k j}\\
&\quad \quad \quad \quad \quad \quad \quad \quad \quad \quad \quad \quad+\sum_{k:~k\in \calN[t_1+1], \,{\bf P}_{i\,k}< {\bf P}_{i^{\prime}\,k}}\pth{{\bf P}_{i k}- {\bf P}_{i^{\prime} k}}\min_{k\in \calN[t_1+1]}{\bf G}_{k j} \Bigg{)}.
\end{align}
Recall that ${\bf P}_{i k}=0={\bf P}_{i^{\prime} k}$ for all $k\notin \calN[t_1+1]$, and $i, i^{\prime} \in \calN[t_1+1]$. We have
\begin{align*}
\sum_{k\in \calN[t_1+1]} {\bf P}_{i k}~=~1~=~\sum_{k\in \calN[t_1+1]} {\bf P}_{i^{\prime} k}.
\end{align*}
Then,
\begin{align*}
0~&=~\sum_{k\in \calN[t_1+1]} {\bf P}_{i k}-\sum_{k\in \calN[t_1+1]} {\bf P}_{i^{\prime} k} \\
&=\sum_{k:~k\in \calN[t_1+1],~{\bf P}_{i k}< {\bf P}_{i^{\prime} k}}\pth{{\bf P}_{i k}- {\bf P}_{i^{\prime} k}}+\sum_{k:~k\in \calN[t_1+1],~{\bf P}_{i k}\ge {\bf P}_{i^{\prime} k}}\pth{{\bf P}_{i k}- {\bf P}_{i^{\prime} k}}.
\end{align*}
Thus, we have
\begin{align}
\label{coeff aux 2}
\sum_{k:~k\in \calN[t_1+1],~{\bf P}_{i k}< {\bf P}_{i^{\prime} k}}\pth{{\bf P}_{i k}- {\bf P}_{i^{\prime} k}}~=~-\sum_{k:~k\in \calN[t_1+1],~{\bf P}_{i k}\ge {\bf P}_{i^{\prime} k}}\pth{{\bf P}_{i k}- {\bf P}_{i^{\prime} k}}.
\end{align}
The RHS of \eqref{coeff aux} becomes
\begin{align}
\nonumber
\delta_{t_1+1}({\bf F})~&\le ~ \max_{j \in \calV}\max_{i, i^{\prime}\in \calN[t_1+1]} \pth{\sum_{k:~k\in \calN[t_1+1], \,{\bf P}_{i\,k}\ge {\bf P}_{i^{\prime}\,k}}\pth{{\bf P}_{i k}- {\bf P}_{i^{\prime} k}}} \pth{\max_{k\in \calN[t_1+1]}{\bf G}_{k j}-\min_{k\in \calN[t_1+1]}{\bf G}_{k j}}\\
\nonumber
&=\max_{i, i^{\prime}\in \calN[t_1+1]} \pth{\sum_{k:~k\in \calN[t_1+1], \,{\bf P}_{i\,k}\ge {\bf P}_{i^{\prime}\,k}}\pth{{\bf P}_{i k}- {\bf P}_{i^{\prime} k}}} \max_{j \in \calV}\max_{k, k^{\prime}\in \calN[t_1+1]}\pth{{\bf G}_{k j}-{\bf G}_{k^{\prime} j}}\\
\nonumber
&=\max_{i, i^{\prime}\in \calN[t_1+1]} \pth{\sum_{k:~k\in \calN[t_1+1], \,{\bf P}_{i\,k}\ge {\bf P}_{i^{\prime}\,k}}\pth{{\bf P}_{i k}- {\bf P}_{i^{\prime} k}}}\delta_{t_1+1}\pth{\bf G}~~~\text{by \eqref{ergodic 1}}\\
\nonumber
&=\max_{i, i^{\prime}\in \calN[t_1+1]} \pth{1-\sum_{k\in \calV}\min \{{\bf P}_{i\,k}, {\bf P}_{i^{\prime}\,k}\}}\delta_{t_1+1}({\bf G})~~~\text{since}\sum_{k\in \calN[t_1+1]} {\bf P}_{i k}=1\\
&=\pth{1-\eta_{t_1+1}({\bf P})}\delta_{t_1+1}({\bf G}),
\end{align}
proving the lemma.

\end{proof}
%
%
%

%
%
%
\vskip 0.4\baselineskip
Recall that $\chi$ is the total number of reduced graphs of $G(\calV, \calE)$. Using Proposition \ref{lower bound}, Lemma \ref{crash erdoc c1} and Lemma \ref{crash erdoc c2}, we can show that $\left | {\bf \Phi}_{ik}(t,r)-{\bf \Phi}_{jk}(t,r)\right |$, where $i, j\in \calN$, decays exponentially fast.
\begin{theorem}
\label{bound}
Suppose $r\le t$. For any $i, j\in \bar{\calN}[t]$ and for any $k\in \calV$, it holds that
\begin{align*}
\left | {\bf \Phi}_{ik}(t,r)-{\bf \Phi}_{jk}(t,r)\right |\le \min\{~1,~ (1-\xi^{n\chi})^{\lfloor \frac{t-r+1}{n\chi}\rfloor -f}\}.
\end{align*}
\end{theorem}
\begin{proof}
For any $i, j\in \bar{\calN}[t]$ and $k\in \calV$, it trivially holds that
\begin{align}
\label{upper bound 1}
\left | {\bf \Phi}_{ik}(t,r)-{\bf \Phi}_{jk}(t,r)\right |\le 1.
\end{align}
To prove this theorem, it is enough to show
\begin{align*}
\left | {\bf \Phi}_{ik}(t,r)-{\bf \Phi}_{jk}(t,r)\right |\le (1-\xi^{n\chi})^{\lfloor \frac{t-r+1}{n\chi}\rfloor -f}, \forall \, i, j\in \bar{\calN}[t].
\end{align*}
For $1\le k\le \lfloor \frac{t-r+1}{n\chi}\rfloor$, let
\begin{align}
\label{block}
{\bf Q}[k]={\bf \Phi}(r+n\chi k-1, r+n\chi (k-1)).
\end{align}
From Proposition \ref{lower bound}, we have
\begin{align}
\label{l1}
\nonumber
{\bf Q}[k]&={\bf A}[r+n\chi k-1]{\bf A}[r+n\chi k-2]\cdots {\bf A}[r+n\chi (k-1)]\\
&\ge~ \xi^{n\chi}\prod_{\tau=r+n\chi (k-1)}^{r+n\chi k-1} {\bf H}[\tau].
\end{align}
The matrix corresponding to
at least one reduced graph, say $\calH^*$, appears at least $n$ times in the product $\prod_{\tau=r+n\chi (k-1)}^{r+n\chi k-1} {\bf H}[\tau]$. Let $S\subseteq \calV$ be the unique source component of $\calH^*$. Due to the existence of self-loops, and the fact that only nodes {\em not} in $\bar{\calN}[\tau]$ can be the removed sink nodes (as per Definition \ref{reduced graph}) in $\calH[\tau]$,  we know that each $j\in S$  reaches every node in $\bar{\calN}[r+n\chi k-1]$. Thus, for $j\in S$
\begin{align}
\label{l2}
\qth{\prod_{\tau=r+n\chi (k-1)}^{r+n\chi k-1} {\bf H}[\tau]}_{ij}\ge ~1, ~\forall\, i\in \bar{\calN}[r+n\chi k-1].
\end{align}
By \eqref{l1} and \eqref{l2}, for  $i\in \bar{\calN}[r+n\chi k-1]$ and $j\in S$, we have
\begin{align}
\label{l3}
{\bf Q}_{ij}[k]~\ge ~\xi^{n\chi}.
\end{align}
Now let us consider the case when $\bar{\calN}[r+n\chi k-1]=\calN[r+n\chi (k-1)]$, that is, no agents crash through iteration $r+n\chi (k-1)$ to the end of iteration $r+n\chi k -1$.
\begin{align}
\label{l4}
\nonumber
&\eta_{r+n\chi (k-1)}\pth{{\bf Q}[k]}=\eta_{r+n\chi (k-1)}\pth{{\bf \Phi}(r+n\chi k-1, r+n\chi (k-1))}\\
\nonumber
&=\min_{i, i^{\prime}\in \calN[r+n\chi (k-1)]}\sum_{j\in \calV} \min \sth{{\bf \Phi}_{ij}\pth{r+n\chi k-1, r+n\chi (k-1)},\, {\bf \Phi}_{i^{\prime}j}\pth{r+n\chi k-1, r+n\chi (k-1)}}\\
\nonumber
&\ge \min_{i, i^{\prime}\in \calN[r+n\chi (k-1)]} \sum_{j\in S} \min \sth{{\bf \Phi}_{ij}\pth{r+n\chi k -1, r+n\chi (k-1)},\, {\bf \Phi}_{i^{\prime}j}\pth{r+n\chi k-1, r+n\chi (k-1)}}\\
\nonumber
&=\min_{i, i^{\prime}\in \bar{\calN}[r+n\chi k -1]} \sum_{j\in S} \min \sth{{\bf \Phi}_{ij}\pth{r+n\chi k-1, r+n\chi (k-1)},\, {\bf \Phi}_{i^{\prime}j}\pth{r+n\chi k-1, r+n\chi (k-1)}}\\
&\ge \min_{i, i^{\prime}\in \bar{\calN}[r+n\chi k -1]}\xi^{n\chi}=\xi^{n\chi}.
\end{align}
%
For $i, i^{\prime}\in \bar{\calN}[t]$, it holds that
\begin{align}
\label{u2}
\nonumber
&\left | {\bf \Phi}_{ij}(t,r)-{\bf \Phi}_{i^{\prime}j}(t,r)\right |\le \max_{j\in \calV} \max_{k,k^{\prime} \in \bar{\calN}[t]} \pth{{\bf \Phi}_{kj}(t, r)-{\bf \Phi}_{k^{\prime}j}(t, r)}\\
\nonumber
&=\max_{j\in \calV} \max_{k,k^{\prime} \in \calN[t+1]} \pth{{\bf \Phi}_{kj}(t, r)-{\bf \Phi}_{k^{\prime}j}(t, r)}~~~\text{since }\bar{\calN}[t]=\calN[t+1]\\
\nonumber
&=\delta_{t+1}\pth{{\bf \Phi}\pth{t, r}} ~~~\text{by \eqref{ergodic 1}}\\
\nonumber
&=\delta_{t+1}\pth{{\bf \Phi}\pth{t, r+n\chi \lfloor \frac{t-r+1}{n\chi}\rfloor}{\bf \Phi}\pth{r+n\chi \lfloor \frac{t-r+1}{n\chi}\rfloor-1, r}}\\
\nonumber
&\le \delta_{r+n\chi \lfloor \frac{t-r+1}{n\chi}\rfloor}\pth{{\bf \Phi}\pth{t, r+n\chi \lfloor \frac{t-r+1}{n\chi}\rfloor}{\bf \Phi}\pth{r+n\chi \lfloor \frac{t-r+1}{n\chi}\rfloor-1, r}} ~~~\text{by \eqref{monotone 1}}\\
\nonumber
&\le \pth{1-\eta_{r+n\chi \lfloor \frac{t-r+1}{n\chi}\rfloor}\pth{{\bf \Phi}\pth{t, r+n\chi \lfloor \frac{t-r+1}{n\chi}\rfloor}}}\delta_{r+n\chi \lfloor \frac{t-r+1}{n\chi}\rfloor} \pth{{\bf \Phi}\pth{r+n\chi \lfloor \frac{t-r+1}{n\chi}\rfloor-1, r}}\\
\nonumber
&\quad \quad \quad \quad \quad \quad \quad \quad \quad \quad \quad \quad \quad \quad \quad \quad\quad \quad \quad \quad \quad \quad \quad \quad
\quad \quad \quad \quad \quad \quad \quad \quad
\text{by Lemma \ref{crash erdoc c2}}\\
&\le \delta_{r+n\chi \lfloor \frac{t-r+1}{n\chi}\rfloor} \pth{{\bf \Phi}\pth{r+n\chi \lfloor \frac{t-r+1}{n\chi}\rfloor-1, r}}.
\end{align}
Similarly, we get
\begin{align}
\label{u31}
\nonumber
&\delta_{r+n\chi \lfloor \frac{t-r+1}{n\chi}\rfloor} \pth{{\bf \Phi}\pth{r+n\chi \lfloor \frac{t-r+1}{n\chi}\rfloor-1, r}}\\
\nonumber
&\le \delta_{r+n\chi \pth{\lfloor \frac{t-r+1}{n\chi}\rfloor-1}} \pth{{\bf \Phi}\pth{r+n\chi \lfloor \frac{t-r+1}{n\chi}\rfloor-1, r}}~~~\text{by \eqref{monotone 1}}\\
\nonumber
&\le \delta_{r+n\chi \pth{\lfloor \frac{t-r+1}{n\chi}\rfloor-1}} \pth{{\bf Q}[\lfloor \frac{t-r+1}{n\chi}\rfloor]{\bf \Phi}\pth{r+n\chi \pth{\lfloor \frac{t-r+1}{n\chi}\rfloor-1}-1, r}}~~~\text{by \eqref{block}}\\
\nonumber
&\le \pth{1-\eta_{r+n\chi \pth{\lfloor \frac{t-r+1}{n\chi}\rfloor-1}}{\bf Q}[\lfloor \frac{t-r+1}{n\chi}\rfloor]}\delta_{r+n\chi \pth{\lfloor \frac{t-r+1}{n\chi}\rfloor-1}} \pth{{\bf \Phi}\pth{r+n\chi \pth{\lfloor \frac{t-r+1}{n\chi}\rfloor-1}-1, r}}\\
&\quad \quad \quad \quad \quad \quad \quad \quad \quad \quad \quad \quad \quad \quad \quad \quad\quad \quad \quad \quad \quad \quad \quad \quad
\quad \quad \quad \quad \quad \quad \quad \quad
\text{by Lemma \ref{crash erdoc c2}}.
\end{align}
Repeatedly apply \eqref{u31}  $\lfloor \frac{t-r+1}{n\chi}\rfloor$ times,  we have
\begin{align}
\label{u3}
\delta_{r+n\chi \lfloor \frac{t-r+1}{n\chi}\rfloor} \pth{{\bf \Phi}\pth{r+n\chi \lfloor \frac{t-r+1}{n\chi}\rfloor-1, r}}\le \prod_{k=1}^{\lfloor \frac{t-r+1}{n\chi}\rfloor}\pth{1-\eta_{r+n\chi (k-1)}\pth{{\bf Q}[k]}},
\end{align}
where each ${\bf Q}[k]$, as defined in \eqref{block}, is a product of $n\chi$ matrices.

Define $k\in \{1, \ldots, \lfloor \frac{t-r+1}{n\chi}\rfloor\}\triangleq \calK$. Let $\bar{\calK}\subseteq \calK$ such that for each $k\in \bar{\calK}$, it holds that $\bar{\calN}[r+n\chi k-1]=\calN[r+n\chi (k-1)]$. Since at most $f$ agents crash, $\left | \bar{\calK}\right |\ge \lfloor \frac{t-r+1}{n\chi}\rfloor -f$. Thus, for $i, i^{\prime}\in \bar{\calN}[t]$ and for $j\in \calV$, we have
\begin{align}
\label{upper bound 2}
\nonumber
\left | {\bf \Phi}_{ij}(t,r)-{\bf \Phi}_{i^{\prime}j}(t,r)\right |&\le
\prod_{k=1}^{\lfloor \frac{t-r+1}{n\chi}\rfloor}\pth{1-\eta_{r+n\chi (k-1)}\pth{{\bf Q}[k]}}~~~\text{by \eqref{u3}}\\
\nonumber
&\le \prod_{k\in \bar{\calK}}\pth{1-\eta_{r+n\chi (k-1)}\pth{{\bf Q}[k]}}\\
\nonumber
&\le \prod_{k\in \bar{\calK}}\pth{1-\xi^{n\chi}}~~~\text{by \eqref{l4}}\\
&\le (1-\xi^{n\chi})^{\lfloor \frac{t-r+1}{n\chi}\rfloor -f}~~~\text{since }\left | \bar{\calK}\right |\ge \lfloor \frac{t-r+1}{n\chi}\rfloor -f
\end{align}

\vskip \baselineskip

Combining \eqref{upper bound 1} and \eqref{upper bound 2}, for $i, i^{\prime}\in \bar{\calN}[t]$ and for $j\in \calV$, we conclude that
\begin{align*}
\left | {\bf \Phi}_{ik}(t,r)-{\bf \Phi}_{jk}(t,r)\right |\le \min\{~1,~ (1-\xi^{n\chi})^{\lfloor \frac{t-r+1}{n\chi}\rfloor -f}\},
\end{align*}
proving the theorem.

\eproof \end{proof}

An immediate consequence of Theorem \ref{bound} is that \eqref{r weak ergodic} holds. Thus Condition 2 is also sufficient for \eqref{r weak ergodic} to hold. \\

Our next lemma states that rows of ${\bf \Phi}(t, r)$ with indices in $\calN$ have a common limit.
\begin{lemma}
\label{lc1}
For all $i\in \calN$, the $i$--th row of ${\bf \Phi}(t, r)$, denoted by ${\bf \Phi}_{i\cdot}(t, r)$ converges to a stochastic vector $\bm{\bm{\pi}}^T(r)$, i.e.,
\begin{align}
\label{limit}
{\bf \Phi}_{i\cdot}(t, r)~\to~ \bm{\bm{\pi}}^T(r), ~\forall ~r.
\end{align}
\end{lemma}

Indeed, for each $i\in \calN$, the row ${\bf \Phi}_{i\cdot} (t, r)$ converges to $\bm{\pi}(r)$ exponentially fast.
\begin{proposition}
\label{lc3}
Suppose $r\le t$. For any $i\in \calN$ and for any $k\in \calV$, it holds that
\begin{align*}
\left | {\bf \Phi}_{ik}(t,r)-\pi_k(r)\right |\le \min\{~1,~ (1-\xi^{n\chi})^{\lfloor \frac{t-r+1}{n\chi}\rfloor -f}\}.
\end{align*}
\end{proposition}
In addition, the limit vector $\bm{\pi}(r)$ has the following property.
\begin{lemma}
\label{lc2}
For any $r$, there exists a reduced graph $\tilde{\calH}[r]$ with source component $S_r$ such that $\pi_j(r)\ge \xi^{n\chi}$, for each $ j\in S_r.$
In addition, $|S_r|\ge \gamma$.
\end{lemma}
The proofs of Lemma \ref{lc1}, Proposition \ref{lc3} and Lemma \ref{lc2} can be found in Appendix \ref{app:lc}.

\section{Convergence of Non-Bayesian Learning}
\label{sec: almost sure}
In the absence of agent failures and messages asynchrony \cite{jadbabaie2012non}, for the networked agents to detect the true hypothesis $\theta^*$, it is enough to assume that $G(\calV, \calE)$ is strongly connected, and that $\theta^*$ is globally identifiable, i.e., for any
$\theta, \theta^*\in \Theta$ and $\theta\not=\theta^*$, there exists a node $j\in \calV$ such that
$$
D \pth{\ell_j(\cdot |\theta^*)||\ell_j(\cdot |\theta)}\triangleq \sum_{w_j\in \calS_j}\ell_j(w_j|\theta^*)\log \frac{\ell_j(w_j|\theta^*)}{\ell_j(w_j|\theta)}~\not=~0,
$$
or equivalently,
\begin{align}
\label{failure-free identify}
\sum_{i\in \calV} D \pth{\ell_i(\cdot |\theta^*)||\ell_i(\cdot |\theta)}~\not=~0.
\end{align}
Since $\theta^*$ may change from execution to execution, \eqref{failure-free identify} is required to hold for any choice of $\theta$ and $\theta^*$ such that $\theta\not=\theta^*$.
Intuitively speaking, if any pair of states $\theta_1$ and $\theta_2$ can be distinguished by at least one agent in the network, then sufficient exchange of local beliefs over strongly connected network will enable every agent to distinguish $\theta_1$ from $\theta_2$. However, in the presence of failures and asynchrony, the effective influence network (reduced graphs) may not be strongly connected. Thus, stronger global identifiability of the network is required.
\begin{assumption}
\label{ass}
For any $\theta, \theta^* \in \Theta$, $\theta\not=\theta^*$, and for any reduced graph $\calH$ of $G(\calV, \calE)$ with $\calS_{\calH}$ denoting the unique source component, the following holds.
\begin{align}
\label{failure identify}
\sum_{i\in \calS_{\calH}} D\pth{\ell_i(\cdot |\theta^*)||\ell_i(\cdot |\theta)}~\not=~0.
\end{align}
\end{assumption}
In contrast to \eqref{failure-free identify}, where the summation is taken over all the agents in the network, in \eqref{failure identify},  the summation is taken over the source component of an arbitrary reduced graph. Intuitively, the condition imposed by Assumption \ref{ass} is that all the agents in the source component can detect the true state $\theta^*$ collaboratively. Then if iterative consensus, specified in \eqref{r weak ergodic}, is achieved, the accurate belief can be propagated from the source component to every non-faulty agent in the network.
Assumption \ref{ass} is necessary for (\ref{goal}) to be achievable. This is because if Assumption \ref{ass} is violated, no information outside the source component is available to the agents in the component to distinguish two hypotheses, and $\theta^*$ cannot be learned.
In addition, Assumption \ref{ass} is also sufficient for $\theta^*$ to be learned. Our convergence analysis of the proposed protocol has a similar structure to that in \cite{nedic2014nonasymptotic}.

Recall that $\calC$ is the collection of all reduced graphs of $G(\calV, \calE)$. Let $\calH\in \calC$ with source component $\calS_{\calH}$.
Define $C_0$ and $C_1$ as
\begin{align}
\label{c0}
-C_0\triangleq  \min_{i\in \calV} \min_{\theta_1, \theta_2\in \Theta; \theta_1\not= \theta_2} \min_{w_i\in \calS_i} \pth{\log \frac{\ell_i(w_i|\theta_1)}{\ell_i(w_i|\theta_2)}},
\end{align}
and
\begin{align}
\label{c1}
C_1\triangleq \min_{\calH\in \calC} \min_{\theta_1, \theta_2\in \Theta; \theta_1\not= \theta_2}\sum_{i\in \calS_{\calH}} D\pth{\ell_i(\cdot |\theta_1)||\ell_i(\cdot |\theta_2)}.
\end{align}
Since $|\Theta|=m< \infty$ and the finiteness of $|\calS_i|$ for each $i\in \calV$, we know that $-C_0>-\infty$. In addition, it is easy to see that $-C_0\le 0$ (thus, $C_0\ge 0$). From Assumption \ref{ass} and the fact that $|\calC|=\chi<\infty$, we get $C_1>0.$

\begin{theorem}
\label{almost sure}
When $G(\calV, \calE)$ satisfies Condition 1, and Assumption 1 holds, each non-faulty agent $i\in \calN$ learns the true hypothesis $\theta^*$ almost surely, i.e., $$\mu_t^i(\theta^*) \toas 1, ~~\forall\, i\in \calN.$$
\end{theorem}
\begin{proof}
Consider any $\theta\in \Theta$ such that $\theta\not=\theta^*$, where $\theta^*$ is the underlying true state. Recall from \eqref{int1} that
\begin{align*}
\bm{\psi}_{t}(\theta)={\bf \Phi}(t,1)\bm{\psi}_0(\theta)+\sum_{r=1}^{t} {\bf \Phi}(t,r+1)\calL_r(\theta)。
\end{align*}
Define $H(\theta)\in \reals^n$ such that
\begin{align}
\label{h}
H_i(\theta)~\triangleq ~ \sum_{w_i\in \calS_i} \ell_i(w_i|\theta^*) \log \frac{\ell_i(w_i|\theta)}{\ell_i(w_i|\theta^*)}=-D \pth{\ell_i(\cdot |\theta^*)||\ell_i(\cdot |\theta)}.
\end{align}
Since $D \pth{\ell_i(\cdot |\theta^*)||\ell_i(\cdot |\theta)}\ge 0$, it follows that $H_i(\theta)\le 0$. \footnote{Alternatively, since $\log (\cdot)$ is concave, by Jensen's inequality, we get \\$H_i(\theta)=\sum_{w_i\in \calS_i} \ell_i(w_i|\theta^*)\log\frac{\ell_i(w_i|\theta)}{\ell_i(w_i|\theta^*)}
\le \log \pth{\sum_{w_i\in \calS} \ell_i(w_i|\theta^*)\frac{\ell_i(w_i|\theta)}{\ell_i(w_i|\theta^*)}}=\log 1=0$.} Note that, for each $t\ge 1$ and for each $i\in \bar{\calN}[t]$, it holds that
\begin{align}
\label{expect}
H_i(\theta) = \mathbb{E}^*[\calL_t^i(\theta)],
\end{align}
where the expectation $\mathbb{E}^*[\cdot]$ is taken over $\ell_i(\cdot |\theta^*)$. Since the support of $\ell_i(\cdot|\theta)$ is the whole signal space $\calS_i$ for each agent $i\in \calV$, it holds that $\left |\frac{\ell_i(w_i|\theta)}{\ell_i(w_i|\theta^*)}\right |<\infty$ for each $w_i\in \calS_i$, and
\begin{align}
\label{finite of H}
\nonumber
H_i(\theta)&= \sum_{w_i\in \calS_i} \ell_i(w_i|\theta^*)\log\frac{\ell_i(w_i|\theta)}{\ell_i(w_i|\theta^*)}\ge \min_{w_i\in \calS_i} \pth{\log \frac{\ell_i(w_i|\theta)}{\ell_i(w_i|\theta^*)}}~~~\text{since }\sum_{w_i\in \calS_i}\ell_i(w_i|\theta^*)=1 \\
&\ge \min_{i\in \calV}\min_{\theta_1, \theta_2\in \Theta; \theta_1\not= \theta_2} \min_{w_i\in \calS_i} \pth{\log \frac{\ell_i(w_i|\theta_1)}{\ell_i(w_i|\theta_2)}} = ~-C_0 >-\infty.
\end{align}
Recall that $\bm{\bm{\pi}}(r+1)\in \reals^n$ is the limit of  $[{\bf \Phi}(t,r+1)]_{i \cdot}$ for each $i\in \calN$, and is stochastic. By \eqref{finite of H}, we know that
\begin{align*}
-\infty <-C_0\le \bm{ \bm{\pi}}^{T}(r+1) H(\theta)\le 0.
\end{align*}
Due to the finiteness of $\bm{ \bm{\pi}}^{T}(r+1) H(\theta)$, we are able to add $\sum_{r=1}^{t} \ones_n \bm{\pi}^{T}(r+1) H(\theta)$ and subtract $\sum_{r=1}^{t} \ones_n \bm{\pi}^{T}(r+1) H(\theta)$ from \eqref{int1}, where $\ones_n\in \reals^n$ is an all one vector of dimension $n$. So we get
\begin{align}
\label{int2}
\nonumber
\bm{\psi}_{t}(\theta)&={\bf \Phi}(t,1)\bm{\psi}_0(\theta)+\sum_{r=1}^{t} {\bf \Phi}(t,r+1)\calL_r(\theta)\\
\nonumber
&={\bf \Phi}(t,1)\bm{\psi}_0(\theta)+\sum_{r=1}^{t} {\bf \Phi}(t,r+1)\calL_r(\theta)-\sum_{r=1}^{t} \ones_n \bm{\pi}^{T}(r+1) H(\theta)+\sum_{r=1}^{t} \ones_n \bm{\pi}^{T}(r+1) H(\theta)\\
&={\bf \Phi}(t,1)\bm{\psi}_0(\theta)+\sum_{r=1}^{t}\pth{{\bf \Phi}(t,r+1)\calL_r(\theta)-\ones_n \bm{\pi}^{T}(r+1) H(\theta)}+\sum_{r=1}^{t} \ones_n \bm{\pi}^{T}(r+1) H(\theta).
\end{align}
For each $i\in \calN$, we have
\begin{align}
\label{evo}
\nonumber
\bm{\psi}^i_{t}(\theta)&=\sum_{k=1}^n{\bf \Phi}_{ik}(t,1)\bm{\psi}^i_0(\theta)
+\sum_{r=1}^{t}\pth{\sum_{k=1}^n{\bf \Phi}_{ik}(t,r+1)\calL^k_r(\theta)- \sum_{k=1}^n\pi_k(r+1)H_k(\theta)}\\
&\quad+\sum_{r=1}^{t}  \sum_{k=1}^n\pi_k(r+1) H_k(\theta).
\end{align}
To show $\lim_{t\diverge}\mu_t^i(\theta^*) \toas 1$, it is enough to show $\bm{\psi}^i_{t}(\theta)\toas -\infty$ for $\theta\not=\theta^*$.
The remaining proof of convergence has similar structure as the analysis in \cite{nedic2014nonasymptotic}.
We will bound the three terms in the right hand side of (\ref{evo}) separately.
Since $\mu_0^i$ is uniform, $\bm{\psi}^i_0(\theta)=0$.
We show that the second term of the RHS of \eqref{evo} decreases linearly in $t$. For each $r\le t$, let $S_r$ be the set of agents that has the property stated in Lemma \ref{lc2}. Then, we have
\begin{align*}
\sum_{k=1}^n\pi_k(r+1) H_k(\theta)&\le \sum_{k\in S_{r+1}}\pi_k(r+1) H_k(\theta)~~~\text{since } H_k(\theta)\le 0\\
&=-\sum_{k\in S_{r+1}}\pi_k(r+1) D(\ell_k(\cdot|\theta^*)||\ell_k(\cdot|\theta))~~~\text{by \eqref{h}}\\
&\le -\xi^{n\chi}\sum_{k\in S_{r+1}}D(\ell_k(\cdot|\theta^*)||\ell_k(\cdot|\theta))~~~\text{by Lemma \ref{lc2}}\\
&\le -\xi^{n\chi}C_1 ~~~\text{by \eqref{c1}}
\end{align*}
Thus, we get
\begin{align}
\label{conv4}
\sum_{r=1}^{t} \sum_{k=1}^n\pi_k(r+1) H_k(\theta)\le-\sum_{r=1}^{t} \sum_{j\in S_r}\xi^{n\chi} |H_j(\theta)|\le -C_1\xi^{n\chi}t.
\end{align}

Using Kolmogorov's strong law of large number, as shown below, we can prove that
\begin{align*}
\frac{1}{t}\sum_{r=1}^{t}\pth{\sum_{k=1}^n{\bf \Phi}_{ik}(t,r)\calL^k_r(\theta)-\sum_{k=1}^n\pi_k(r+1) H_k(\theta)} \toas 0.
\end{align*}
Specifically,
\begin{align}
\label{second}
\nonumber
&\frac{1}{t}\sum_{r=1}^{t}\pth{\sum_{k=1}^n{\bf \Phi}_{ik}(t,r+1)\calL^k_r(\theta)- \sum_{k=1}^n\pi_k(r+1)H_k(\theta)}\\
\nonumber
&=\frac{1}{t}\sum_{r=1}^{t}\pth{\sum_{k=1}^n{\bf \Phi}_{ik}(t,r+1)\calL^k_r(\theta)-
\sum_{k=1}^n \pi_k(r+1)\calL^k_r(\theta)+\sum_{k=1}^n \pi_k(r+1)\calL^k_r(\theta)- \sum_{k=1}^n\pi_k(r+1)H_k(\theta)}\\
&=\frac{1}{t}\sum_{r=1}^{t}\sum_{k=1}^n\pth{{\bf \Phi}_{ik}(t,r+1)-\pi_k(r+1)}\calL^k_r(\theta)+\frac{1}{t}\sum_{r=1}^{t}\pth{\sum_{k=1}^n \pi_k(r+1)\calL^k_r(\theta)- \sum_{k=1}^n\pi_k(r+1)H_k(\theta)}.
\end{align}
For each $k\notin \bar{\calN}[t]$, by \eqref{private signal}, we know that
\begin{align}
\label{d1}
\left |\calL^k_r(\theta)\right |=0.
\end{align}
For each $k\in \bar{\calN}[t]$, by \eqref{private signal}, we have
\begin{align*}
\left |\calL^k_r(\theta)\right |=\left |\log \frac{\ell_i(s_{t}^i|\theta)}{\ell_i(s_{t}^i|\theta^*)}\right |\le \max_{i\in \calV} \max_{\theta_1, \theta_2 \in \Theta; \theta_1\not= \theta_2}
\max_{w_i\in \calS_i} \left |\log \frac{\ell_i(w_i|\theta_1)}{\ell_i(w_i|\theta_2)}\right |.
\end{align*}
Note that $\max_{i\in \calV} \max_{\theta_1, \theta_2 \in \Theta; \theta_1\not= \theta_2}
\max_{w_i\in \calS_i}  \left |\log \frac{\ell_i(w_i|\theta_1)}{\ell_i(w_i|\theta_2)}\right |$ is symmetric in $\theta_1$ and $\theta_2$. Thus,
\begin{align}
\label{d2}
\nonumber
\left |\calL^k_r(\theta)\right |&\le \max_{i\in \calV} \max_{\theta_1, \theta_2 \in \Theta; \theta_1\not= \theta_2}
\max_{w_i\in \calS_i}  \left |\log \frac{\ell_i(w_i|\theta_1)}{\ell_i(w_i|\theta_2)}\right |=\max_{i\in \calV} \max_{\theta_1, \theta_2 \in \Theta; \theta_1\not= \theta_2}
\max_{w_i\in \calS_i}  \log \frac{\ell_i(w_i|\theta_1)}{\ell_i(w_i|\theta_2)}\\
\nonumber
&=\max_{i\in \calV} \max_{\theta_1, \theta_2 \in \Theta; \theta_1\not= \theta_2}
\max_{w_i\in \calS_i}  -\log \frac{\ell_i(w_i|\theta_2)}{\ell_i(w_i|\theta_1)}\\
&=-\min_{i\in \calV} \min_{\theta_1, \theta_2 \in \Theta; \theta_1\not= \theta_2}
\min_{w_i\in \calS_i} \log \frac{\ell_i(w_i|\theta_2)}{\ell_i(w_i|\theta_1)}=-(-C_0)=C_0.
\end{align}
By \eqref{d1} and \eqref{d2}, we get
\begin{align}
\label{d3}
\left |\calL^k_r(\theta)\right |\le C_0, ~~~ \forall \, k\in \calN.
\end{align}
Note that $\sum_{r=1}^{t} \min\{1, ~(1-\xi^{n\chi})^{\lfloor \frac{t-r}{n\chi}\rfloor-f}\}$ is a geometric series for sufficient large $t$, and is convergent. In particular,
\begin{align}
\label{d4}
\sum_{r=1}^{t} \min\{1, ~(1-\xi^{n\chi})^{\lfloor \frac{t-r}{n\chi}\rfloor-f}\}\le \sum_{r=1}^{\infty} \min\{1, ~(1-\xi^{n\chi})^{\lfloor \frac{t-r}{n\chi}\rfloor-f}\}~\triangleq~ C<\infty.
\end{align}
For each $i\in \calN$, the absolute value of the first term in the right hand side of (\ref{second}) can be bounded as follows.
\begin{align}
\label{E}
\nonumber
&\frac{1}{t}\left |\sum_{r=1}^{t} \sum_{k=1}^n \pth{{\bf \Phi}_{ik}(t,r+1)-\pi_k(r+1)}\calL^k_r(\theta)\right |\\
\nonumber
&\le\frac{1}{t}\sum_{r=1}^{t} \sum_{k=1}^n \left |{\bf \Phi}_{ik}(t,r+1)-\pi_k(r+1) \right | \left |\calL^k_r(\theta)\right |\\
\nonumber
&\le \frac{1}{t}\sum_{r=1}^{t} \sum_{k=1}^n \left |{\bf \Phi}_{ik}(t,r+1)-\pi_k(r+1) \right | C_0~\text{by \eqref{d3}}\\
\nonumber
&=\frac{1}{t}\pth{\sum_{r=1}^{t-1} \sum_{k=1}^n \left |{\bf \Phi}_{ik}(t,r+1)-\pi_k(r+1) \right | C_0+\sum_{k=1}^n \left |{\bf \Phi}_{ik}(t,t+1)-\pi_k(t+1) \right | C_0}\\
\nonumber
&\le \frac{1}{t}\pth{\sum_{r=1}^{t-1} \sum_{k=1}^n \left |{\bf \Phi}_{ik}(t,r+1)-\pi_k(r+1) \right | C_0+n C_0}~~~\text{since $\left |{\bf \Phi}_{ik}(t,r+1)-\pi_k(r+1) \right |\le 1$}\\
\nonumber
&\le \frac{1}{t}\pth{\sum_{r=1}^{t-1} n \cdot\min\{1, ~(1-\xi^{n\chi})^{\lfloor \frac{t-r}{n\chi}\rfloor-f}\} C_0 +n C_0} ~~~\text{by Proposition \ref{lc3}}\\
\nonumber
&\le \frac{1}{t}\sum_{r=1}^{t} n \cdot\min\{1, ~(1-\xi^{n\chi})^{\lfloor \frac{t-r}{n\chi}\rfloor-f}\} C_0~~~\text{since }\min\{1, ~(1-\xi^{n\chi})^{\lfloor \frac{t-t}{n\chi}\rfloor-f}\}=1\\
&\le \frac{1}{t}nCC_0 ~~~\text{by \eqref{d4}}
\end{align}
Taking limsup on both sides of (\ref{E}), we get
\begin{align*}
 \limsup_{t\diverge} \frac{1}{t}\left |\sum_{r=1}^{t} \sum_{k=1}^n \pth{{\bf \Phi}_{ik}(t,r+1)-\pi_k(r+1)}\calL^k_r(\theta)\right |& \le \limsup_{t\diverge} \frac{1}{t}nCC_0 \\
 &=0\\
 &\le  \liminf_{t\diverge} \frac{1}{t}\left |\sum_{r=1}^{t} \sum_{k=1}^n \pth{{\bf \Phi}_{ik}(t,r+1)-\pi_k(r+1)}\calL^k_r(\theta)\right |.
\end{align*}
Thus, the limit of $\frac{1}{t}\left |\sum_{r=1}^{t} \sum_{k=1}^n \pth{{\bf \Phi}_{ik}(t,r+1)-\pi_k(r+1)}\calL^k_r(\theta)\right |$ exists, and
\begin{align}
\label{conv1}
\lim_{t\diverge}\frac{1}{t}\left |\sum_{r=1}^{t} \sum_{k=1}^n \pth{{\bf \Phi}_{ik}(t,r+1)-\pi_k(r+1)}\calL^k_r(\theta)\right |~~=~~0.
\end{align}
Therefore,
\begin{align}
\label{limit tt}
\lim_{t\diverge}\frac{1}{t} \sum_{r=1}^{t} \sum_{k=1}^n \pth{{\bf \Phi}_{ik}(t,r+1)-\pi_k(r+1)}\calL^k_r(\theta) = 0,
\end{align}
i.e., the first term in the right hand side of \eqref{second} goes to 0.

Now we bound the second term in the right hand side of (\ref{second}), using Kolmogorov's strong law of large number.
Kolmogorov's strong law of large number states that if $\{X_r\}_{r=1}^{\infty}$ is a sequence of {\em independent} random variables such that $\sum_{r=1}^{\infty} \frac{Var(X_r)}{r^2}<\infty$, then
$$\frac{1}{t}\sum_{r=1}^t X_k-\frac{1}{t}\sum_{r=1}^t \mathbb{E}[X_r] ~\toas ~0.$$

For each $r$, let
\begin{align}
\label{def x}
X_r\triangleq \sum_{k=1}^n \pi_k(r+1) \calL_r^k (\theta).
\end{align}
Since the private signals $s_r^i$'s are i.i.d. across iterations, $X_r$'s are independent.
In addition, since
\begin{align*}
\expect{X_r^2}=\expect{\pth{\sum_{k=1}^n \pi_k(r+1) \calL_r^k (\theta)}^2}\le \expect{\max_{k\in \calV} |\calL_r^k (\theta)|^2}\le C_0^2,
\end{align*}
where the last inequality follows from \eqref{d3}, we get
\begin{align*}
\sum_{r=1}^{\infty} \frac{Var\pth{X_r}}{r^2}  \le \sum_{r=1}^{\infty} \frac{\expect{X_r^2}}{r^2}\le \sum_{r=1}^{\infty} \frac{1}{r^2} C_0^2 <\infty.
\end{align*}
For each $i\in \calN$ and each $k\in \calV$, by Lemma \ref{lc1}, it holds that
$\pi_k(r+1) =\lim_{t\diverge} {\bf \Phi}_{ik}(t, r+1).$ In addition,
by \eqref{p3}, for each $k\notin \calN[t+1]$, we have
\begin{align}
\label{cc0}
\pi_k(r+1) =\lim_{t\diverge} {\bf \Phi}_{ik}(t, r+1)=0.
\end{align}
Thus, we get
\begin{align}
\label{ccc}
X_r = \sum_{k=1}^n \pi_k(r+1) \calL_r^k (\theta)=\sum_{k\in \calN[t+1]} \pi_k(r+1) \calL_r^k (\theta).
\end{align}

Thus, we have, for each $i\in \calN$
\begin{align}
\label{conv2}
\nonumber
\frac{1}{t}\sum_{r=1}^{t}\pth{\sum_{k=1}^n \pi_k(r+1)\calL^k_r(\theta)- \sum_{k=1}^n\pi_k(r+1)H_k(\theta)}&=\frac{1}{t}\sum_{r=1}^{t}\pth{X_r-  \sum_{k=1}^n\pi_k(r+1)H_k(\theta)}~~~\text{by \eqref{def x}}\\
\nonumber
&=\frac{1}{t}\sum_{r=1}^{t}\pth{X_r-  \sum_{k\in \calN[r+1]}\pi_k(r+1)H_k(\theta)}~~~\text{by \eqref{cc0}}\\
\nonumber
&=\frac{1}{t}\sum_{r=1}^{t}\pth{X_r- \sum_{k\in \calN[r+1]}\pi_k(r+1)\mathbb{E}^*[\calL_r^k(\theta)]}~\text{by \eqref{expect}}\\\nonumber
&=\frac{1}{t}\sum_{r=1}^{t}\pth{X_r- \mathbb{E}^*[X_r]}~\text{by \eqref{ccc}} \\
&\toas ~0.
\end{align}
By \eqref{limit tt}, \eqref{conv2} and \eqref{second}, we obtain that
\begin{align}
\label{conv3}
\frac{1}{t}\sum_{r=1}^{t}\pth{\sum_{k=1}^n{\bf \Phi}_{ik}(t,r)L^k_r(\theta)-\sum_{k=1}^n\pi_k(r+1) H_k(\theta)} \toas 0.
\end{align}

\vskip 2\baselineskip
From \eqref{evo}, for each $i\in \calN$ and $\theta\not=\theta^*$, we get
\begin{align*}
\limsup_{t\diverge}\frac{1}{t} \bm{\psi}_{t}^i(\theta)&=\limsup_{t\diverge}{\bf \Phi}_{ik}(t,1)\bm{\psi}_0^i(\theta)+
\limsup_{t\diverge}\frac{1}{t} \sum_{r=1}^{t}\pth{\sum_{k=1}^n{\bf \Phi}_{ik}(t,r+1)\calL^k_r(\theta)- \sum_{k=1}^n\pi_k(r+1)H_k(\theta)}\\
&\quad+\frac{1}{t} \limsup_{t\diverge}\sum_{r=1}^{t}  \sum_{k=1}^n\pi_k(r+1) H_k(\theta)\\
&=0+\limsup_{t\diverge}\frac{1}{t} \sum_{r=1}^{t}\pth{\sum_{k=1}^n{\bf \Phi}_{ik}(t,r+1)\calL^k_r(\theta)- \sum_{k=1}^n\pi_k(r+1)H_k(\theta)}~~~\text{since $\bm{\psi}^i_0(\theta)=0$}\\
&\quad+ \limsup_{t\diverge}\frac{1}{t}\sum_{r=1}^{t}  \sum_{k=1}^n\pi_k(r+1) H_k(\theta)\\
&\le \lim_{t\diverge}\frac{1}{t}\sum_{r=1}^{t}\pth{\sum_{k=1}^n{\bf \Phi}_{ik}(t,r)\calL^k_r(\theta)-\sum_{k=1}^n\pi_k(r+1) H_k(\theta)}-\liminf_{t\diverge}\pth{\frac{1}{t}t\xi^{n\chi}C_1} ~~~\text{by \eqref{conv4}}\\
&=-\xi^{n\chi}C_1 ~~~\text{by \eqref{conv3}} ~~~\text{a.s.}
\end{align*}
Consequently, for each $i\in \calN$ and each $\theta\not=\theta^*$, we have
\begin{align*}
 \bm{\psi}_{t}^i(\theta)\toas -\infty. ~~~\text{and}~~~\lim_{t\diverge}\mu_t^i(\theta) \toas 0.
\end{align*}
Due to \eqref{b1} and the fact that $\tilde{\mu}_t^i(\theta)=\mu_t^i(\theta), \forall i\in \calN$, we get
\begin{align*}
\lim_{t\diverge}\mu_t^i(\theta) \toas 0.
\end{align*}

Thus, for each $i\in \calN$,
\begin{align*}
\lim_{t\diverge}\mu_t^i(\theta^*) \toas 1.
\end{align*}

Therefore, we have shown that all non-faulty agents learn the true parameter $\theta^*$ almost surely.

\eproof
\end{proof}

\begin{remark}
As it can be seen from the above analysis, our results can be trivially adapted to time-varying graphs. In addition, most of the existing results on (asymptotic and finite-time) convergence rates can be generalized to our setting.
\end{remark}
\newpage
\appendix

\setlength {\parskip}{6pt}

\centerline{\Large\bf Appendices}

\section{Equivalence of Condition 1 and Condition 2}
\label{app:equivalence}
The proof of Theorem \ref{equivalence} relies on the following proposition.
\begin{proposition}
\label{unique source}
Every subgraph $\calH^{\prime}$ of $G(\calV, \calE)$ obtained by removing up to $f$ incoming links for each $i\in \calV$ contains a unique source component.
\end{proposition}
\begin{proof}
Suppose $\calH^{\prime}$ contains at least two source components. Let $L$ and $R$ be two source components in $\calH^{\prime}$, and let $C=\calV-L-R$. It is easy to see that $L\not=\O$ and $R\not=\O$. Since $L$ is a source component, no nodes in $L$ have incoming links from $R\cup C$. In addition, by the construction of $\calH^{\prime}$, we know that in $G(\calV, \calE)$, every node in $L$ has at most $f$ incoming neighbors in $R\cup C$, i.e., $|N_i^-\cap (R\cup C)|\le f$ for each $i\in L$. Similarly, we have
$|N_j^-\cap (L\cup C)|\le f$ for each $j\in R$. This contradicts the assumption that $G(\calV, \calE)$ satisfies Condition 1.

Therefore, $\calH^{\prime}$ contains only one source component.

\eproof
\end{proof}

\begin{proof}[Proof of Theorem \ref{equivalence}]
We first show that Condition 1 implies Condition 2, i.e., if the graph $G(\calV, \calE)$ satisfies Condition 1, then by Proposition \ref{unique source}, every reduced graph $\calH$ of $G(\calV, \calE)$, defined as per Definition \ref{reduced graph}, has only one source component.

Let $\calH^{\prime}$ be an arbitrary subgraph of $G(\calV, \calE)$ obtained by removing up to $f$ incoming links for each node. From Definition \ref{reduced graph}, it can be seen that to show Condition 1 implies Condition 2, it is enough to show that every reduced graph obtained from $\calH^{\prime}$ contains only one source component.

From Definition \ref{reduced graph}, we know that when $\calH^{\prime}$ contains no sink nodes, $\calH^{\prime}$ is the only reduced graph associated with $\calH^{\prime}$. By Proposition \ref{unique source}, $\calH^{\prime}$ contains only one source component.
Similarly, from Definition \ref{reduced graph}, when $\calH^{\prime}$ contains at least one sink node -- a node that does not have any outgoing links, multiple reduced graphs can be obtained from $\calH^{\prime}$. Let $\calH$ be an arbitrary reduced graph associated with $\calH^{\prime}$ such that $\calH \not=\calH^{\prime}$, with $\{i_1, \ldots, i_p\}$ being the $p$ (where $1\le p\le f$) sink nodes removed from $\calH^{\prime}$.
The uniqueness of source component in $\calH^{\prime}$ implies that there exists a node that can reach any other node in $\calH^{\prime}$. That is, there exists a node $j$ such that for any $i\in \calV$, $i\not=j$, there is a $j,i$--path in $\calH^{\prime}$.
Since a sink node does not have any outgoing link, any intermediate node of a $j,i$--path in $\calH^{\prime}$ cannot be a sink node. Then, if $i\notin \{i_1, \ldots, i_p\}$, any $j,i$--path in $\calH^{\prime}$ is remained in $\calH$.
Thus, node $j$ can reach any node $i\notin \{i_1, \ldots, i_p\}$ in $\calH$, proving $\calH$ contains only one source component. By the arbitrariness of $\calH$, we conclude that every reduced graph obtained from $\calH^{\prime}$ contains only one source component.

Since $\calH^{\prime}$ is also chosen arbitrarily, therefore, we conclude that Condition 1 implies Condition 2.  \\

Next we show that Condition 2 implies Condition 1, i.e., if every reduced graph of $G(\calV, \calE)$ contains only one source component, then $G(\calV, \calE)$ satisfies Condition 1. We prove this by contradiction.
Suppose $G(\calV, \calE)$ satisfies Condition 2, but does not satisfy Condition 1, i.e., there exists a node partition $L, R, C$ such that $L\not=\O$, $R\not=\O$, and the following holds:
for $i\in L$, $|N_i^-\cap \pth{R\cup C}|\le f$, and for $j\in R$, $|N_j^-\cap \pth{L\cup C}|\le f$.
Now, consider the reduced graph constructed as follows: for each $i\in L$, remove all incoming links from $R\cup C$, and for each $j\in R$, remove all incoming links form $L\cup C$. Denote the obtained subgraph as $\calH$. From Definition \ref{reduced graph}, we know that $\calH$ is a reduced graph. Since nodes in $L$ do not have any incoming link from $R\cup C$, and nodes in $R$ do not have any incoming link from $L\cup C$, thus both $L$ and $R$ contain source components. Consequently, there are at least two source components in $\calH$, contradicting the fact that $\calH$ contains only one source component. Thus, we conclude that Condition 2 implies Condition 1. \\

Therefore, Condition 1 and Condition 2 are equivalent.

\eproof \end{proof}

\section{Proof of Proposition \ref{lower bound}}
\label{app:lower bound}
\begin{proof}
We consider two cases separately: (1) every agent participates in iteration $t$, i.e., $\calN[t]=\calV$; and (2) there exists an agent that does not participate in iteration $t$, i.e., $\calV-\calN[t]\not=\O$.
\vskip 0.5\baselineskip

\noindent{\bf Case (1):} Suppose every agent participates in iteration $t$, i.e., $\calN[t]=\calV$. By the construction of matrix ${\bf A}[t]$ in (\ref{matrix 1})  and (\ref{matrix 2}), for each $i\in \calV$, ${\bf A}_{ij}[t]\ge \xi$ for $|N_i^-|-f+1$ incoming neighbors of node $i$.
This corresponds to the operation of removing $f$ incoming links of node $i$.
Thus, there exists a reduced graph $\calH[t]$ with adjacency matrix ${\bf H}[t]$ such that
$${\bf A}[t]\ge \xi {\bf H}[t].$$

\noindent{\bf Case (2):} Suppose there exists agent that does not participate iteration $t$, i.e., $\calV-\calN[t]\not=\O$. In this case, since each agent in $\calV-\calN[t]$ has already crashed during the first $t-1$ iterations, it will no longer send out messages as well as receiving messages. That is, nodes in $\calV-\calN[t]$ are isolated nodes in iteration $t$. By construction of ${\bf A}[t]$ in (\ref{matrix 1})  and (\ref{matrix 2}), we know that  ${\bf A}_{ij}[t]=0$ for each $j\in \calV-\calN[t]$, where $i\not=j$.
Intuitively speaking, this corresponds to the fact that a crashed agent cannot influence any other agent, and thus is a sink node.
In addition, the fact that ${\bf A}_{ji}[t]=0$ for any $j\not=i$ and $j\in \calV-\calN[t]$ corresponds to the operation that the sink nodes in $\calV-\calN[t]$ are removed. Therefore, there exists
a reduced graph $\calH[t]$ with adjacency matrix ${\bf H}[t]$ such that
$${\bf A}[t]\ge \xi {\bf H}[t].$$

Therefore, the proof of Proposition \ref{lower bound} is complete.

\eproof
\end{proof}

\section{Proofs of Lemma \ref{lc1}, Proposition \ref{lc3} and Lemma \ref{lc2}}
\label{app:lc}
\begin{proof}[Proof of Lemma \ref{lc1}]
Since $\reals$ is complete, it is enough to show that $\{{\bf \Phi}_{ij}(t, r)\}_{t=r}^{\infty}$ is a Cauchy sequence. That is, we need to show that for any $\epsilon>0$, there exists $t^*\ge r$ such that for any $t\ge t^*$ and any $p\in \naturals$
$$\left | {\bf \Phi}_{ij}(t+p, r) - {\bf \Phi}_{ij}(t, r)  \right | <\epsilon ,~~\forall \, i\in \calN, \forall\, j\in \calV.$$

By Theorem \ref{bound}, we know for any $\epsilon>0$, there exists $t^*\ge r$ such that for all $t\ge t^*$
\begin{align}
\label{cc1}
\left | {\bf \Phi}_{ij}(t,r)-{\bf \Phi}_{kj}(t,r)\right |<\epsilon, ~\forall \,  i, k\in \bar{\calN}[t], \forall\, j\in \calV.
\end{align}
For any $p\in \naturals$, we have
\begin{align*}
\left | {\bf \Phi}_{ij}(t+p,r)-{\bf \Phi}_{ij}(t,r)\right |&=\left |\sum_{k=1}^n{\bf \Phi}_{ik}(t+p,t+1){\bf \Phi}_{kj}(t,r) -{\bf \Phi}_{ij}(t,r)\right |\\
&=\left |\sum_{k=1}^n{\bf \Phi}_{ik}(t+p,t+1)\pth{{\bf \Phi}_{kj}(t,r) -{\bf \Phi}_{ij}(t,r)}\right |~~~\text{since }\sum_{k=1}^n{\bf \Phi}_{ik}(t+p,t+1)=1\\
&=\left |\sum_{k\in \calN[t+1]}{\bf \Phi}_{ik}(t+p,t+1)\pth{{\bf \Phi}_{kj}(t,r) -{\bf \Phi}_{ij}(t,r)}\right |~~~\text{by Proposition \ref{p3}}\\
&\le \sum_{k\in \calN[t+1]}{\bf \Phi}_{ik}(t+p,t+1)\left|{\bf \Phi}_{kj}(t,r) -{\bf \Phi}_{ij}(t,r)\right |\\
&=\sum_{k\in \bar{\calN}[t]}{\bf \Phi}_{ik}(t+p,t+1)\left|{\bf \Phi}_{kj}(t,r) -{\bf \Phi}_{ij}(t,r)\right |\\
&\le \pth{\sum_{k\in \bar{\calN}[t]}{\bf \Phi}_{ik}(t+p,t+1)}\cdot \epsilon~~~\text{by \eqref{cc1}}\\
&=\epsilon.
\end{align*}
Thus, there exists $\bm{\bm{\pi}}(r)$ such that \eqref{limit} holds. Since ${\bf \Phi}(t, r)$ is row-stochastic, the limiting vector $\bm{\bm{\pi}}(r)$ is also stochastic, proving the lemma.

\eproof
\end{proof}

\begin{proof}[Proof of Proposition \ref{lc3}]
We now show that $ {\bf \Phi}_{ik}(t,r)-\pi_k(r)\le \min\{~1,~ (1-\xi^{n\chi})^{\lfloor \frac{t-r+1}{n\chi}\rfloor -f}\},$ for $i\in \calN$ and $k\in \calV$.
For any $p\in \naturals$,
\begin{align}
\nonumber
{\bf \Phi}_{ik}(t,r)-\pi_k(r)&={\bf \Phi}_{ik}(t,r)-\lim_{p\diverge }{\bf \Phi}_{ik}(t+p,r)~~~\text{by Lemma \ref{lc1}}\\
\nonumber
&=\lim_{p\diverge }\pth{{\bf \Phi}_{ik}(t,r)-{\bf \Phi}_{ik}(t+p,r)}\\
\nonumber
&=\lim_{p\diverge }\pth{{\bf \Phi}_{ik}(t,r)-\sum_{j=1}^n{\bf \Phi}_{ij}(t+p,t+1){\bf \Phi}_{jk}(t, r) }\\
\nonumber
&=\lim_{p\diverge }\sum_{j=1}^n{\bf \Phi}_{ij}(t+p,t+1)\pth{{\bf \Phi}_{ik}(t,r)-{\bf \Phi}_{jk}(t, r) }~~~\text{since }\sum_{j=1}^n{\bf \Phi}_{ij}(t+p,t+1)=1\\
\nonumber
&=\lim_{p\diverge }\sum_{j\in \bar{\calN}[t]}{\bf \Phi}_{ij}(t+p,t+1)\pth{{\bf \Phi}_{ik}(t,r)-{\bf \Phi}_{jk}(t, r) }\\
\nonumber
&\quad+\lim_{p\diverge }\sum_{j\notin \bar{\calN}[t]}  {\bf \Phi}_{ij}(t+p,t+1)\pth{{\bf \Phi}_{ik}(t,r)-{\bf \Phi}_{jk}(t, r) }  \\
&=\lim_{p\diverge }\sum_{j\in \bar{\calN}[t]}{\bf \Phi}_{ij}(t+p,t+1)\pth{{\bf \Phi}_{ik}(t,r)-{\bf \Phi}_{jk}(t, r) }~~~\text{by Proposition \ref{p3}} \label{XX}\\
\nonumber
&\le \sum_{j\in \bar{\calN}[t]}{\bf \Phi}_{ij}(t+p,t+1)\cdot \min\{~1,~ (1-\xi^{n\chi})^{\lfloor \frac{t-r+1}{n\chi}\rfloor -f}\}~~~\text{by Theorem \ref{bound} and $\bar{\calN}[t]=\calN[t+1]$}\\
\nonumber
&=\min\{~1,~ (1-\xi^{n\chi})^{\lfloor \frac{t-r+1}{n\chi}\rfloor -f}\}~~~\text{by \eqref{sum}}
\end{align}
Similarly, we get
\begin{align*}
{\bf \Phi}_{ik}(t,r)-\pi_k(r)&=\lim_{p\diverge }\sum_{j\in \bar{\calN}[t]}{\bf \Phi}_{ij}(t+p,t+1)\pth{{\bf \Phi}_{ik}(t,r)-{\bf \Phi}_{jk}(t, r) }~~~\text{as per \eqref{XX}}\\
&\ge \sum_{j\in \bar{\calN}[t]}{\bf \Phi}_{ij}(t+p,t+1)\cdot \pth{-\min\{~1,~ (1-\xi^{n\chi})^{\lfloor \frac{t-r+1}{n\chi}\rfloor -f}\}}~~~\text{by Theorem \ref{bound}}\\
&=-\min\{~1,~ (1-\xi^{n\chi})^{\lfloor \frac{t-r+1}{n\chi}\rfloor -f}\}.
\end{align*}

The proof of Proposition \ref{lc3} is complete.
\eproof
\end{proof}

In addition, the limit vector $\bm{\pi}(r)$ has the following property.
\begin{proof}[Proof of Lemma \ref{lc2}]
For each $j \in \calV$, we have
\begin{align*}
\pi_j(r)&=\lim_{t\diverge} {\bf \Phi}_{ij}(t,r)\\
&=\lim_{t\diverge} \sum_{k=1}^n{\bf \Phi}_{ik}(t,r+n\chi) {\bf \Phi}_{kj}(r+n\chi-1, r)\\
&=\lim_{t\diverge} \pth{\sum_{k\in \calN[r+n\chi]}{\bf \Phi}_{ik}(t,r+n\chi) {\bf \Phi}_{kj}(r+n\chi-1, r)+\sum_{k\notin \calN[r+n\chi]}{\bf \Phi}_{ik}(t,r+n\chi) {\bf \Phi}_{kj}(r+n\chi-1, r)}\\
&=\lim_{t\diverge} \sum_{k\in \calN[r+n\chi]}{\bf \Phi}_{ik}(t,r+n\chi) {\bf \Phi}_{kj}(r+n\chi-1, r).
\end{align*}
From the proof of Theorem  \ref{bound}, we know that there exists a reduced graph $\tilde{\calH}[r]$ with source component $S_r$ such that for each $j\in S_r$,
$${\bf \Phi}_{kj}(r+n\chi-1, r)\ge \xi^{n\chi}, ~~~\forall \, k \in \bar{\calN}[r+n\chi -1].$$
Thus, for each $j\in S_r$ we have
\begin{align*}
\pi_j(r)&=\lim_{t\diverge} \sum_{k\in \calN[r+n\chi]}{\bf \Phi}_{ik}(t,r+n\chi) {\bf \Phi}_{kj}(r+n\chi-1, r)\\
&\ge \lim_{t\diverge} \pth{\sum_{k\in \calN[r+n\chi]}{\bf \Phi}_{ik}(t,r+n\chi) }~\xi^{n\chi}\\
& =\xi^{n\chi} ~~~\text{by Proposition \ref{p3} }.
\end{align*}
Recall that $\gamma$ is the minimum source component size in all reduced graph. So $|S_r|\ge \gamma$. The proof is complete.

\eproof
\end{proof}

\bibliographystyle{abbrv}

\end{document}